\pdfoutput=1
\documentclass[11pt]{article}
\usepackage{fullpage}
\usepackage{mathtools,amsfonts,amsthm,amssymb,bbm}
\usepackage{xcolor}
\usepackage[backref=section,colorlinks,citecolor=blue,linkcolor=magenta,bookmarks=true]{hyperref}
\usepackage[nameinlink]{cleveref}
\usepackage{authblk}

\title{Online Pen Testing%
\thanks{We would like to thank Ian Tullis, Petr Mitrichev, and the entire problem setting team of Google Code Jam 2020 for writing and preparing the problem titled \emph{Pen Testing}~\cite{Tullis20}, which inspired this work. We thank the anonymous reviewers for their comments that have helped improve this paper. This work was supported by NSF awards 1813049, 1704417 and 1804222, and DOE award DE-SC0019205.}
}

\date{}

\ifdefined\anon
    \author{Anonymous Authors}
\else
    \author{Mingda Qiao}
    \author{Gregory Valiant}
    \affil{\texttt{\{mqiao,valiant\}@stanford.edu}}
    \affil{Stanford University}
\fi

\newcommand{\1}[1]{\mathbbm{1}\left[#1\right]} 

\newcommand{\D}{\mathcal{D}} 
\newcommand{\eps}{\epsilon}
\newcommand{\E}{\mathcal{E}} 
\newcommand{\Ex}[2]{\operatorname*{\mathbb{E}}_{#1}\left[#2\right]} 
\newcommand{\Exp}{\mathrm{Exp}}

\newcommand{\pr}[2]{\Pr_{#1}\left[#2\right]} 
\newcommand{\R}{\mathbb{R}} 
\newcommand{\rmd}{\mathrm{d}}

\newcommand{\Xmax}{X^{\textrm{max}}}
\newcommand{\Xmaxhat}{\widehat{a_{[1]}}}

\newtheorem{theorem}{Theorem}
\newtheorem{definition}{Definition}
\newtheorem{lemma}{Lemma}[section]
\newtheorem{fact}[lemma]{Fact}

\newtheorem{remark}[lemma]{Remark}

\begin{document}

\maketitle

\begin{abstract}
    We study a ``pen testing'' problem, in which we are given $n$ pens with unknown amounts of ink $X_1, X_2, \ldots, X_n$, and we want to choose a pen with the maximum amount of \emph{remaining} ink in it. The challenge is that we cannot access each $X_i$ directly; we only get to write with the $i$-th pen until either a certain amount of ink is used, or the pen runs out of ink. In both cases, this testing reduces the remaining ink in the pen and thus the utility of selecting it.
    
    Despite this significant lack of information, we show that it is possible to approximately maximize our utility up to an $O(\log n)$ factor. Formally, we consider two different setups: the ``prophet'' setting, in which each $X_i$ is independently drawn from some distribution $\D_i$, and the ``secretary'' setting, in which $(X_i)_{i=1}^n$ is a random permutation of arbitrary $a_1, a_2, \ldots, a_n$.
    We derive the optimal competitive ratios in both settings up to constant factors. Our algorithms are surprisingly robust: (1) In the prophet setting, we only require \emph{one} sample from each $\D_i$, rather than a full description of the distribution; (2) In the secretary setting, the algorithm also succeeds under an \emph{arbitrary} permutation, if an estimate of the maximum $a_i$ is given.
    
    Our techniques include a non-trivial online sampling scheme from a sequence with an unknown length, as well as the construction of a hard, non-uniform distribution over permutations. Both might be of independent interest. We also highlight some immediate open problems and discuss several directions for future research.
\end{abstract}

\section{Introduction}
Suppose that we have a few whiteboard pens to choose from for an upcoming  presentation. We want to maximize the amount of remaining ink in the pen we pick, measured in writing time. Naturally, before we make our decision, we write with each pen for a short while to check whether the ink has almost run out. We face a dilemma regarding how long each pen should be tested. If we use the pen for just five seconds, we could not distinguish whether it had ten seconds or twenty minutes of writing time at the beginning. At the other extreme, too long a test period may exhaust the ink in the pen, leaving us too little ink for the actual writing.

This toy problem models scenarios such as testing the service life of a flimsy spare part, and more generally,  other real-world decision-making in which obtaining information about each option inevitably reduces the utility of the option. For example, we want to invest in one of $n$ start-ups with unknown growth potentials. We could, of course, watch from the sidelines for a while, and see whether the total value of each company has grown to a certain point (e.g., twice its initial value). However, we would have a lower return since our investment in the company only starts at this point, and the 2x increase does not count towards our profit.

To our knowledge, this ``pen testing'' problem first appeared as a competitive programming problem, written by Ian Tullis and prepared by Petr Mitrichev, in Google Code Jam 2020~\cite{Tullis20}. They considered the case that $n = 15$ pens hold $0, 1, 2, \ldots, n-1$ units of ink respectively, but are presented to us after a random shuffling. We can test the pens in an arbitrary order, and possibly go back to a pen that we tested earlier for further testing, if this is deemed necessary. Finally, we are asked to choose two of the pens, and we win the game if the total units of remaining ink in them is at least $n$.

If we randomly pick two pens without any testing, our winning probability is clearly below $50\%$. Surprisingly, it was shown by~\cite{Tullis20} that at $n = 15$, a better strategy wins the game with a higher probability of $\approx 64.4\%$! This strategy is computed by dynamic programming, in which each state simply consists of all the information that we obtain from testing: the amount of ink that has been used from each pen, and whether each pen has run out or not. For larger $n$, however, this approach would necessarily result in an exponential runtime. Furthermore, in the general case that the amounts of ink in the pens are no longer a permutation of $(0, 1, 2, \ldots, n-1)$, it is difficult to analyze how the optimal solution computed by this dynamic programming scales asymptotically.

In this work, we formulate and study an online version of this pen testing problem, in which we select only one of the $n$ options, and both the testing and decision are subject to an additional temporal restriction---we must review the $n$ options in the given order. For each option, we are allowed to test it to some extent, during which the value of the option also decreases. Then, we need to make an irrevocable decision on whether to accept the option---once we accept, the entire game ends and we cannot explore the remaining options; once we reject an option, we can no longer go back to it if the later options appear less ideal.

This \emph{online pen testing} problem that we consider is closely related to the theory of optimal stopping, in which the player is often assumed to observe the value of each option directly. In the single-choice case that we focus on, two well-studied settings are the \emph{prophet inequality} and the \emph{secretary problem}. In the former, the values are assumed to be drawn independently from $n$ given distributions. In the secretary problem, the options can have arbitrary values but the $n$ options are assumed to arrive in a uniformly random order. We discuss the connection between our work and this literature in Section~\ref{sec:related}. In this paper, we study the online pen testing problem in settings similar to these two problems, and derive the optimal guarantee that the player can achieve under minimal assumptions on the option values.

Awerbuch, Azar, Fiat, and Leighton~\cite{AAFL96} studied a closely related and more general setting: A decision-maker may hold at most one of $n$ commodities on each day. At the end of the day, each commodity issues a dividend of either $0$ or $1$ to its holders. The goal is to achieve a total profit comparable to the dividend issued by the best commodity, by switching between the commodities as few times as possible. In Remark~\ref{remark:AAFL} we discuss how pen testing can be realized as a special case of this setting.  We discuss this connection further when describing our results (Section~\ref{sec:results}) and techniques (Section~\ref{sec:techniques}).

\subsection{Problem Setup}\label{sec:problem-setup}
We first define the online pen testing problem formally.

\begin{definition}[Online pen testing]
    A problem instance is specified by $X_1, X_2, \ldots, X_n \ge 0$. At each step $i \in [n]$, the player first tests $X_i$ and then makes a decision:
    \begin{itemize}
        \item \textbf{(Testing)} The player picks threshold $\theta_i \in [0, +\infty]$. If $X_i > \theta_i$, the test passes; the test fails if $X_i \le \theta_i$, in which case the player observes $X_i$. The remaining utility of option $i$ becomes $X'_i = \max\{X_i - \theta_i, 0\}$.
        
        \item \textbf{(Decision)} After seeing whether the test passes, the player either accepts or rejects the $i$-th option irrevocably. If the player accepts, the game ends and the player receives a score of $X'_i$.
    \end{itemize}
\end{definition}
\begin{remark}
    The player may pick threshold $\theta_i = +\infty$, in which case the player gets to observe $X_i$ at the cost of leaving a remaining utility of $X'_i = 0$.
\end{remark}

\begin{remark}
    Our definition allows a more general testing procedure, in which the player performs $t$ tests sequentially at chosen thresholds $\theta^{(1)}_i, \theta^{(2)}_i, \ldots, \theta^{(t)}_i \ge 0$. This is equivalent to running a single test at threshold $\theta_i = \theta^{(1)}_i + \theta^{(2)}_i + \cdots + \theta^{(t)}_i$.
\end{remark}

\begin{remark}\label{remark:AAFL}
The problem can be viewed as a special case of the setting studied by~\cite{AAFL96}, where no switching is allowed and the commodities issue their dividend sequentially. Assuming that $X_1, X_2, \ldots, X_n$ are all integers, the online pen testing problem corresponds to an instance where the $i$-th commodity issues a unit dividend on $X_i$ consecutive days starting from day number $1 + \sum_{j=1}^{i-1}X_j$, and zero dividend on each of the other days.
\end{remark}

This problem can be viewed as a variant of the well-studied optimal stopping problem in which information is both limited and costly. For each option $i$, we either: (1) receive a single bit of information (namely, that $X_i > \theta_i$ holds) at the cost of reducing the value of the option by $\theta_i$; or: (2) observe $X_i$ exactly when $X_i \le \theta_i$, at the cost of losing all the utility in option $i$.

Without any assumptions on $X_1, \ldots, X_n$, no non-trivial guarantee on the player's score can be made.%
\footnote{This is true even if the player can observe the value $X_i$ directly.} 
In this paper, we consider the following two setups: the ``prophet'' setting and the ``secretary'' setting, both of which make some distributional assumption on the instance. In the following, we formally define the settings and the notion of competitive ratio in each of them.

\begin{definition}[Prophet setting]
    The player is given information about distributions $\D_1, \D_2, \ldots, \D_n$ over $[0, +\infty)$, from which the $n$ values $X_1, X_2, \ldots, X_n$ are drawn independently.
    
    The player is $\alpha$-competitive if its expected score, over the randomness in the distributional information, the generation of $(X_i)_{i=1}^n$, and the player itself, is at least $\frac{1}{\alpha}\cdot\Ex{X \sim \D}{\max_{i \in [n]}X_i}$.
\end{definition}

\begin{remark}\label{remark:distribution}
Formally, each $\D_i$ is defined by a cumulative distribution function $F_i: \R \to [0, 1]$ that is non-decreasing, right-continuous, and satisfies $\lim_{x\to+\infty}F_i(x) = 1$ and $F_i(x) = 0$ for $x <0$. The resulting $\D_i$ satisfies $\pr{X_i \sim \D_i}{X_i \le x} = F_i(x)$. We assume that each $\D_i$ has a finite expectation, i.e., the integral $\int_0^{+\infty}[1 - F_i(x)]~\rmd x$ converges, which implies that the expected maximum, $\Ex{X \sim \D}{\max_{i \in [n]}X_i}$, is also finite.

We will sometimes assume for simplicity that each $\D_i$ is continuous, i.e., the corresponding $F_i(x)$ is continuous. In other words, $\D_i$ has no point masses. The general case can be handled using a simple reduction (see e.g.,~\cite{RWW20}). The continuity of $F_i$ guarantees that for any $\alpha \in (0, 1]$, we may define the $(1 - \alpha)$-quantile of $\D_i$ as the minimum number $\tau$ that satisfies $\pr{X_i \sim \D_i}{X_i > \tau} = 1 - F_i(\tau) = \alpha$.
\end{remark}

In the prophet setting, the values of different options are drawn independently from distributions $\D_1, \D_2, \ldots, \D_n$. At the beginning of the game, the player is given certain information about the distributions.%
\footnote{Without information about $(\D_i)_{i=1}^n$, this is as hard as the worst-case setting when every $\D_i$ is degenerate.}
We consider both the case that the player receives a complete description of $(\D_i)_{i=1}^n$, and the case where the player sees one sample $\hat X_i$ drawn from each $\D_i$. In the latter case, the observed sample $\hat X_i$ is independent from the actual value $X_i$, and the expected score of the player is defined over the randomness in $(\hat X_i)_{i=1}^n$ as well. Finally, the player's score is compared to that of an omniscient prophet that knows the realization of $X_1, \ldots, X_n$ and always picks the highest one.

\begin{definition}[Secretary setting]
    The player is given information about $a_1, a_2, \ldots, a_n \ge 0$. The $n$ values $X_1, X_2, \ldots, X_n$ are guaranteed to be a permutation, either uniformly random or arbitrary, of $a_1, a_2, \ldots, a_n$.
    
    The player is $\alpha$-competitive in the random order case if its expected score, over the randomness in both the permutation and the player itself, is at least $\frac{1}{\alpha}\cdot\max_{i\in[n]}a_i$. The player is $\alpha$-competitive in the arbitrary order case if for any permutation $(X_i)_{i=1}^n$ of $(a_i)_{i=1}^n$, the player's expected score, over the randomness in the player itself, is at least $\frac{1}{\alpha}\cdot\max_{i\in[n]}a_i$.
\end{definition}

Note that an $\alpha$-competitive player for the arbitrary order case is also $\alpha$-competitive under a random arrival order. We consider the following three forms of information provided to the player, in decreasing order of helpfulness: (1) \emph{full information}, the player is given $a_1, a_2, \ldots, a_n$; (2) \emph{optimum information}, the player is given $\max_{i\in[n]}a_i$; (3) \emph{no information}, the player is given nothing.

With full or optimum information, if the player could observe each $X_i$ directly, it would be easy to achieve a utiliy of $\max_{i \in [n]}a_i$---simply accept option $i$ only if $X_i$ is equal to this maximum. This is, however, not true for the online pen testing problem. For example, when $(a_1, a_2, \ldots, a_n) = (1, 2, \ldots, n)$, the player can only ensure that $X_i = n$ by setting the threshold $\theta_i$ to $n - 1$, but this would leave a remaining utility of merely $1$.

\subsection{Our Results}\label{sec:results}
We obtain the optimal competitive ratios (modulo constant factors) for online pen testing, under different variants of the prophet and secretary settings defined above. 

\paragraph{A simple lower bound.} The following example shows that even when the $X_i$'s are drawn independently from the same ``nice'' distribution, our score can still be an $\Omega(\log n)$ factor away from the optimal outcome.

\begin{fact}\label{fact:lower}
    Suppose that $X_1, X_2, \ldots, X_n$ are drawn independently from the exponential distribution with parameter $1$. The expected score of the player is at most $1$, while the maximum among $X_1, X_2, \ldots, X_n$ is $H_n \coloneqq \sum_{k=1}^{n}\frac{1}{k} = \Omega(\log n)$ in expectation.
\end{fact}

\begin{proof}[Proof of Fact~\ref{fact:lower}]
    Whenever the player accepts option $i$ after testing it at threshold $\theta_i$, the expected remaining utility is $\Ex{X \sim \D}{X|X > \theta_i} - \theta_i = 1$.
    The expected score of the player is thus at most $1$. The second claim follows from a straightforward calculation, which is deferred to Appendix~\ref{sec:omitted-prophet-iid}.
\end{proof}

Note that the instance above is a special case of the prophet setting, in which $\D_1, \D_2, \ldots, \D_n$ are the same distribution. Furthermore, the lower bound argument still goes through even in the ``offline'' setting of~\cite{Tullis20}, i.e., the player is allowed to: (1) test the pens in an arbitrary order; (2) come back and test some pen that has been tested; (3) accept any pen after gathering all the information.

Perhaps surprisingly, this $\Omega(\log n)$ lower bound is the only obstacle against a competitive algorithm: An $O(\log n)$-competitive algorithm exists in almost all the variants, even though the player is under an additional temporal restriction, and has far less information about $(X_i)_{i=1}^n$.

\paragraph{The prophet setting.} Our first result addresses the prophet setting, assuming that the player is given full descriptions of the distributions.

\begin{theorem}\label{thm:prophet}
    In the prophet setting, there is an algorithm that, given $\D_1, \D_2, \ldots, \D_n$, achieves a competitive ratio of $O(\log n)$.
\end{theorem}

In light of Fact~\ref{fact:lower}, the $O(\log n)$ competitive ratio is tight up to a constant factor. This positive result can be strengthened to an $O(\log n)$-competitive single-sample prophet inequality.

\begin{theorem}\label{thm:prophet-single-sample}
    In the prophet setting, there is an algorithm that, given samples $\hat X_1, \hat X_2, \ldots, \hat X_n$ independently drawn from $\D_1, \D_2, \ldots, \D_n$, achieves a competitive ratio of $O(\log n)$.
\end{theorem}

\paragraph{The secretary setting.}
Our positive result for the secretary setting states that an $O(\log n)$ competitive ratio is achievable even if we are given no information about $(a_i)_{i=1}^n$ (under a random arrival order). Furthermore, if the maximum value $\max_{i \in [n]}a_i$ is given, an $O(\log n)$-competitive algorithm exists even if the arrival order is arbitrary. The $O(\log n)$ upper bound for the latter case also follows from a result of~\cite{AAFL96} and the reduction outlined in Remark~\ref{remark:AAFL}. Interestingly, we obtain this competitive ratio using a quite different approach; we compare these two algorithms in more detail in Section~\ref{sec:techniques}.

\begin{theorem}[Secretary setting, upper bounds; Theorem 2.3~of~\cite{AAFL96}]\label{thm:secretary-upper}
    In the secretary setting, an $O(\log n)$-competitive algorithm exists in the following two cases: (1) the order is random and the player is given no information; (2) the order is arbitrary and the player is given optimum information.
\end{theorem}

We prove a matching $\Omega(\log n)$ lower bound under settings that are even easier than those in Theorem~\ref{thm:secretary-upper}.

\begin{theorem}[Secretary setting, lower bounds]\label{thm:secretary-lower}
    In the secretary setting, any algorithm is at best $\Omega(\log n)$-competitive in the following two cases: (1) the order is random and the player is given optimum information; (2) the order is arbitrary and the player is given full information.
\end{theorem}

In the easiest combination among all secretary settings---that $(X_i)_{i=1}^n$ is a random permutation of known values $(a_i)_{i=1}^n$, the competitive ratio is slightly improved to $\Theta\left(\frac{\log n}{\log\log n}\right)$.

\begin{theorem}\label{thm:secretary-improved}
    In the secretary setting with random order and full information, there is an $O\left(\frac{\log n}{\log\log n}\right)$-competitive algorithm. Furthermore, this is tight up to a constant factor.
\end{theorem}

We found this result particularly surprising: if each $a_i$ were drawn from an exponential distribution, the instance appears similar to the lower bound instance of Fact~\ref{fact:lower} and it may seem hard to achieve a super-constant improvement over this. The key insight is that, given the random ordering of the $a_i$'s, with good probability, there will be some $j$ such that the set $\{X_i: i > j\}$ contains a super-constant ``gap'' in the following sense: There exists an interval $[A,B]$ of length $B - A > \Omega(\log \log n)$ such that $\{X_i: i > j\} \cap [A,B] = \emptyset$, and for some $i > j$,  $X_i > B$.  Given this, a simple scheme can achieve a score of $B-A.$  The core of the proof of the $O\left(\frac{\log n}{\log \log n}\right)$ upper bound is showing that for any set of $a_i$'s, such a gap will exist with good probability over the random order.

We summarize the results for the secretary setting in Table~\ref{table:secretary}. The rows represents the arrival order of $a_1, a_2, \ldots, a_n$, and the columns represents the amount of information about $(a_i)_{i=1}^n$ that is provided to the player. The setting becomes harder (or remains equally hard) going from top to bottom and from left to right in the table.

\begin{table}[ht]
    \centering
    \caption{A summary of results for the secretary setting. The top-left cell follows from Theorem~\ref{thm:secretary-improved}. The bottom-right cell is folklore. The remaining four bounds follow from Theorems \ref{thm:secretary-upper}~and~\ref{thm:secretary-lower}.}
    \label{table:secretary}
    
    \vspace{11pt}
    
    \begin{tabular}{|c|c|c|c|}
    \hline &&&\\[-0.5em]
        & Full Information  & Optimum Information & No Information\\
    [-0.5em]&&&\\
    \hline &&&\\[-0.5em]
    Random Order & $\Theta\left(\frac{\log n}{\log\log n}\right)$ & $\Theta(\log n)$  & $\Theta(\log n)$\\
    [-0.5em]&&&\\
    \hline &&&\\[-0.5em]
    Arbitrary Order & $\Theta(\log n)$  & $\Theta(\log n)$  & $\Theta(n)$\\
    [-0.5em]&&&\\
    \hline
    \end{tabular}
\end{table}

\subsection{Proof Overview and Technical Highlights}\label{sec:techniques}
We sketch the proofs of all our results, and highlight a few technical difficulties that are tackled using new techniques that might be of independent interest.

\paragraph{The prophet setting, given the distributions.}
In the i.i.d.\ case that $\D_1 = \D_2 = \cdots = \D_n = \D$, we show that the following \emph{single-threshold algorithm} succeeds: (1) Pick $\theta \ge 0$ and test each option with the same $\theta_i = \theta$; (2) Accept the first option that passes the test.
In particular, we prove that the single-threshold algorithm at one of the thresholds among $\{\tau_1, \tau_{1/2}, \tau_{1/4}, \tau_{1/8}, \ldots, \tau_{1/n}\}$ is $O(\log n)$-competitive, where $\tau_\alpha$ is the $(1 - \alpha)$-quantile of $\D$.

In the general case, however, this single-threshold approach no longer works: If each $\D_i$ is the degenerate distribution at value $i$, any single-threshold algorithm gives a score of at most $1$, whereas the optimum is $n$. Nevertheless, in this problem instance, our knowledge of $(\D_i)_{i=1}^n$ should allow us to realize that option $n$ is the optimal one. Indeed, our algorithm for the general case uses the distributional knowledge to identify a set of ``valuable'' options and only test those options at carefully chosen thresholds (which depend on the individual distributions).


We remark that our analysis for the i.i.d.\ setting (especially the use of the exponential distribution in the lower bound) resembles part of the analysis in~\cite{FTTZ16}, albeit in the different context of money burning auction introduced by~\cite{HR08}.

\paragraph{The secretary setting, upper bounds.} When the arrival order is random and the optimum $a_{[1]} \coloneqq \max_{i \in [n]}a_i$ is known, the single-threshold approach again gives an $O(\log n)$-competitive algorithm. We randomly choose a threshold $\theta$ between $0$ and $a_{[1]}$. It is easy to prove that with probability  $\Omega(1)$, among all the options with value $> \theta$, more than half of them have values higher than $\theta + \Omega(a_{[1]} / \log n)$. The random arrival order then implies that the first option that passes the test at $\theta$ leaves an $\Omega(a_{[1]} / \log n)$ remaining value in expectation.

To prove Theorem~\ref{thm:secretary-upper}, we need to remove either the knowledge of $a_{[1]}$, or the assumption on the arrival order. The former can be done by estimating $a_{[1]}$ using a standard technique. To handle an arbitrary arrival order, however, turns out to be non-trivial. Intuitively, among the $n_{>\theta}$ options with value $> \theta$, we want to accept one of them uniformly at random, but this is difficult without knowing $n_{>\theta}$ in advance.

We define a ``bit sampling'' game that abstracts this challenge. 
\begin{quote} \textbf{Bit Sampling Game:} We observe an arbitrary sequence of $m$ bits, with the promise that strictly more than half of the bits are ``1''.   Crucially, \emph{we do not know $m$ in advance}. We see the bits one by one, and may choose to commit to the next unseen bit at any point. We win if our chosen bit is a ``$1$''.  Can we win with constant probability? Can we win with probability $\ge 1/2$?
\end{quote}


Natural approaches to the problem (e.g., by guessing the value of $m$) only win the game with probability $O(1/\log n)$, where $n$ is an upper bound on $m$. In Section~\ref{sec:bit-sampling}, we give a more intricate strategy that wins with probability $\Omega(1)$. This strategy then gives an algorithm for Case~(2) of Theorem~\ref{thm:secretary-upper}. 

Drucker~\cite{Dru13} studied a similar sampling problem, in which the bit sequence is infinite and the density of 1's is lower bounded asymptotically. (Formally, the average of the first $n$ bits has a limit inferior of $\ge 1 - \eps$ as $n \to +\infty$.) One main result of~\cite{Dru13} is a family of strategies that ``commit to a bit $1$'' with a probability arbitrarily close to $1 - \eps$. Despite the similarity between them, the two problems have different cruxes. Intuitively, the player in the setting of~\cite{Dru13} needs to wait patiently for the sequence to reach a ``high-density region'', whereas the player in the above game must commit more aggressively, in case that the sequence ends very early.

As mentioned earlier, the $O(\log n)$ upper bound for the arbitrary arrival order case also follows from Theorem~2.3 in~\cite{AAFL96}. When translated into the setting of online pen testing, their algorithm picks a randomized threshold for each option, and accepts the first option that passes the test. The thresholds are independently drawn from a flipped exponential distribution over $[0, a_{[1]}]$, i.e., the probability of picking a higher threshold is exponentially larger. In contrast, our approach uses the same, uniformly random threshold for all options, and then uses the ``bit sampling'' scheme to ensure the competitive ratio.

\paragraph{The secretary setting, lower bounds.} The proof for the first case of Theorem~\ref{thm:secretary-lower} (random order and optimum information) follows from a change-of-distribution argument that transforms a distribution over instances to another distribution that corresponds to the prophet setting. The proof for the other case (arbitrary order and full information) is relatively more difficult.

Recall that the prophet setting lower bound (Fact~\ref{fact:lower}) relies on the memoryless property of the exponential distribution. 
It is thus natural to consider a sequence $a_1, a_2, \ldots, a_n$ that contains $\approx n/2$ copies of $1$, $\approx n/4$ copies of $2$, $\approx n/8$ copies of $3$, $\ldots$, and exactly one occurrence of $\log_2 n$, since the uniform distribution over $\{a_1, a_2, \ldots, a_n\}$ is roughly a geometric distribution, which is also memoryless. Hence, if $(X_i)_{i=1}^n$ is a random permutation of $(a_i)_{i=1}^n$, no matter how the player tests the first option, the expected remaining value is at most $O(1)$. If the same were true for all the remaining options in the sequence, an $\Omega(\log n)$ lower bound would follow.

However, this argument does not work perfectly---the construction only gives a weaker lower bound of $\Omega\left(\frac{\log n}{\log\log n}\right)$ in Theorem~\ref{thm:secretary-improved}.%
\footnote{In fact, this is inevitable in light of the upper bound part of Theorem~\ref{thm:secretary-improved}.} Note that the player, given full information about $(a_i)_{i=1}^n$, knows the multiset of the unseen values $\{X_i, X_{i+1}, \ldots, X_n\}$ at any step $i$. If the uniform distribution over this set does not ``resemble a geometric'' for some $i$, the player might exploit this to achieve a super-constant score. 
As we prove in the upper bound part of Theorem~\ref{thm:secretary-improved}, under a random arrival order and regardless of the choice of $(a_i)_{i=1}^n$, this ``non-geometric'' property holds at some point $i$ with a decent probability, so the player can always shave a $\log\log n$ factor off the competitive ratio.

In Section~\ref{sec:non-uniform-order}, we prove the $\Omega(\log n)$ lower bound in Theorem~\ref{thm:secretary-lower}, Case (2) by constructing a more intricate distribution over permutations of essentially the same sequence $(a_i)_{i=1}^n$. This distribution ensures that w.h.p.\ \emph{every} suffix of the sequence $(X_i)_{i=1}^n$ resembles the geometric distribution, and thus the player can achieve an $O(1)$ score at best.

\paragraph{The single-sample prophet setting.}
Our proof of Theorem~\ref{thm:prophet} implies that, to be $O(\log n)$-competitive in the prophet setting, it suffices to know a few quantiles of $\D_1$ through $\D_n$. We might hope that the same algorithm can be implemented using the samples. However, this approach would not prove Theorem~\ref{thm:prophet-single-sample}, since the algorithm needs the $(1-1/n)$-quantile of each $\D_i$, which requires $\Omega(n)$ samples from each distribution to estimate.

Interestingly, our results for the secretary setting can be applied to prove the single-sample prophet inequality in Theorem~\ref{thm:prophet-single-sample}. Given the samples $\hat X_1, \hat X_2, \ldots, \hat X_n$, we use $\Xmaxhat \coloneqq \max_{i \in [n]}\hat X_i$ as an estimate for the maximum among the ``real values'' $X_1, \ldots, X_n$. The problem instance can then be viewed as the arbitrary-order, optimum-information case of the secretary setting, except that we only know a rough ``hint'' on the maximum. Fortunately, our algorithm for Theorem~\ref{thm:secretary-upper}, Case~(2) can handle this case as well.

\subsection{Related Work}\label{sec:related}
Our work is closely related to the vast literature on the prophet inequality introduced by Krengel, Sucheston and Garling~\cite{KS78} and the secretary problem that dates back to at least the work of Dynkin~\cite{Dyn63}. We refer the readers to a tutorial of~\cite{Gup17} for different solutions for these two problems.

\paragraph{Prophets and secretaries with costs.} Most closely related to this paper is the prior work on optimal stopping with observation costs. For prophet inequalities, Jones~\cite{Jon90} first considered a setting where the player has to pay a fixed cost of $c \ge 0$ to observe the value of each item. In other words, the net reward from accepting the $i$-th option is reduced to $X_i - ic$. \cite{Jon90} derived sharp bounds on the \emph{difference} between the score of the optimal player and that of a prophet who knows $X_1, \ldots, X_n$. A special case of this setting that the values are i.i.d.\ was subsequently studied by~\cite{SC92,Har96,Kos04}.

Bartoszy{\'n}ski and Govindarajulu~\cite{BG78} defined a variant of the secretary problem with ``interview costs''. Given constants $a, b, c_1, c_2, \ldots, c_n \ge 0$, the player pays a cost of $c_k$ if the $k$-th option is selected. Furthermore, a score of $a$ or $b$ is awarded, depending on whether the chosen option has the highest or second highest value. More recently, \cite{BDGI09} studied a similar secretary problem with discounts, in which the value of the $i$-th option is $d(i)\cdot X_i$, where $X_1, X_2, \ldots, X_n$ are arbitrary values that arrive in a random order, and $d(\cdot)$ is a given discount function.

Another related setting is the Pandora's Box problem first introduced by Weitzman~\cite{Wei79}. In this setting, the value $v_i$ of each option is independently drawn, and the player may choose to examine the $i$-th option (i.e., to learn $v_i$) at a posted cost of $c_i$. Recent work has studied several variants of Pandora's Box: multiple selection under a combinatorial constraint~\cite{Sin18}, with a more complex probing process specified by Markov chains~\cite{GJSS19}, with a correlated prior distribution~\cite{CGTTZ20}, or under additional restrictions on the order of probing~\cite{BFLL20}.

In all these previous settings, the player can still fully access the value of each option, and the observation cost mostly depends on the number of options that the player observes before accepting. In contrast, our model assumes a more restricted form of observation and a cost that is commensurate with the extent to which we observe each option. On the other hand, the player only needs to pay the cost for the option that it finally accepts.

\paragraph{Multiple selection under constraints.} While we focus on the single choice setting of optimal stopping, we remark that there has been a flurry of recent work in the TCS community on selecting multiple options, in either prophet or secretary setting, under certain combinatorial constraints~\cite{Kle05,BIK07,CL12,Sot13,JSZ13,DK14,Lac14,FSZ14,MTW16,Rub16,RS17,HN20,STV21,AL21,AKKG21}. In particular, for the \emph{matroid secretary problem} in which the player is required to choose an independent set from some given matroid with the $n$ options as the ground set, it remains a major open problem whether an $O(1)$-competitive algorithm exists.

\paragraph{Prophet inequalities from samples.} Whereas most prior work on prophet inequalities makes the arguably strong assumption on the full knowledge of the distributions, \cite{AKW14} explored the setting where the player only observes a few samples from each distribution, and proved several $O(1)$-competitive prophet inequalities that require only a single sample from each distribution, even for the multiple-selection case under several types of matroid constraints. For the single-choice i.i.d.\ case, \cite{CDFS19} gave an algorithm that achieves a competitive ratio of $\alpha + \eps$ with $O(n^2)$ samples for any constant $\eps > 0$, where $\alpha \approx \frac{1}{0.745}$ is the optimal competitive ratio when the distribution is known. \cite{RWW20} further improved the sample complexity to $O(n)$ for the i.i.d.\ case, and also gave a $2$-competitive single-sample prophet inequality for the non-i.i.d.\ case. The competitive ratio of $2$ matches the case that the distributions are given.

\subsection{Organization of the Paper}
We start by highlighting several natural open problems and discussing future directions of research in Section~\ref{sec:discussion}. In Section~\ref{sec:prophet-iid}, we present a simple $O(\log n)$-competitive algorithm for the prophet setting under an additional i.i.d.\ assumption. This is extended to the general non-i.i.d.\ case (Theorem~\ref{thm:prophet}) in Section~\ref{sec:prophet-general}. Algorithms and lower bound constructions for the secretary setting (Theorems \ref{thm:secretary-upper}~through~\ref{thm:secretary-improved}) are presented in Sections \ref{sec:secretary-upper}~and~\ref{sec:secretary-lower}. In Section~\ref{sec:prophet-single-sample}, we show how the results for the secretary setting imply the single-sample prophet inequality in Theorem~\ref{thm:prophet-single-sample}.

\section{Discussion and Open Problems}\label{sec:discussion}
We mention a few immediate open problems and directions for future work.

\paragraph{Optimal constants.} The main focus of this work is the order of the competitive ratio, so we will prioritize the clarity of the algorithms/hard instances over optimizing the constant factors. (That said, all the hidden constants in our results will be reasonably small and easy to keep track of.) The most immediate open problem is to pin down the optimal constant factors in the competitive ratios as $n$ grows. Even in the i.i.d.\ case of the prophet setting, a multiplicative gap of $e$ still exists: The lower bound from Fact~\ref{fact:lower} scales as $H_n = (1 + o(1))\ln n$, whereas in Appendix~\ref{sec:constants}, we show that the upper bound from Theorem~\ref{thm:prophet} can be refined to $(e + o(1))\ln n$.

We note that in the usual setup of single-choice prophet inequalities, the optimal competitive ratio of $2$ for the non-i.i.d.\ case was obtained in the fundamental work of Krengel, Sucheston and Garling~\cite{KS78}, while the i.i.d.\ case turned out to be more challenging, and was solved only very recently by~\cite{CFH+17}.

\paragraph{Multiple selection and matroids.} A natural extension of the current single-choice setting is to select several pens, while maximizing the total remaining ink in them. The selection is subject to a cardinality constraint or, more generally, an arbitrary combinatorial constraint. Following the seminal work of~\cite{BIK07,BIKK18}, there has been a flurry of recent work on matroid secretary problem, in which the combinatorial structure is a matroid, and it remains an open problem to bound the optimal competitive ratio in terms of the matroid rank (the current best bounds are $\Omega(1)$ and $O(\log\log\mathrm{rank})$).

For many special types of matroids, $O(1)$-competitive algorithms are known. In particular, as noted in~\cite{BDGI09}, many of such algorithms proceed by reducing a more general matroid to a partition matroid, and the problem essentially becomes multiple instances of the single-choice problem. Therefore, any matroid class that has such a partition property also admits an $O(\log n)$-competitive algorithm in the pen testing variant. It remains an interesting question whether a similar $O(\log n)$ competitive ratio can be achieved in online pen testing over other natural matroid types. Conversely, does the lack of information in pen testing make it easier to prove lower bounds?

\paragraph{More general cost functions.} In our setting, testing an option with threshold $\theta$ reduces its utility by $\theta$. We could consider a slightly more general setting where the remaining utility of option $i$ becomes $X_i - c(\theta_i)$ for some given cost function $c(\cdot)$. For example, in the special case that $c(\theta_i) = \alpha \theta_i$ for some $\alpha > 0$, natural extensions of our current algorithmic techniques can still be applied.

\paragraph{A smoothed feedback.} Currently, the feedback from a test is binary---either ``pass'' or ``fail''. In the whiteboard pen example that motivates this work, it is also realistic to assume that we can tell when the ink ``starts to run out''. It would be interesting to formulate such a setting and explore whether this smoothed feedback makes stronger guarantees possible.

\section{Warmup: The IID Prophet Setting}\label{sec:prophet-iid}
We start with a special case of the prophet setting that $\D_1, \D_2, \ldots, \D_n$ are the same distribution (denoted by $\D$), a full description of which is given to the player. For simplicity, we assume that $\D$ is a continuous probability distribution. By Remark~\ref{remark:distribution}, for $\alpha \in (0, 1]$, we can define $\tau_\alpha$ as the smallest $(1-\alpha)$-quantile of $\D$, i.e., the minimum $\tau$ such that $\pr{X \sim \D}{X > \tau} = \alpha$.

In the following, $X_1, \ldots, X_n$ are independent random variables that follow distribution $\D$, and $\Xmax \coloneqq \max\{X_1,\ldots,X_n\}$ is defined as their maximum. We start by upper bounding the expected optimum, $\Ex{X_1,\ldots,X_n\sim\D}{\Xmax}$:

\begin{lemma}\label{lemma:prophet-optimum}
    For any distribution $\D$ and $\theta \in \R$,
    \[\Ex{X_1,\ldots,X_n\sim\D}{\Xmax} \le \theta + n\cdot\Ex{X\sim\D}{\max\{X-\theta, 0\}}.\]
\end{lemma}
\begin{proof}
    We have
    \begin{align*}
        \Ex{X_1,\ldots,X_n\sim\D}{\Xmax}
    &\le \theta + \Ex{X_1,\ldots,X_n\sim\D}{\max\{\Xmax - \theta, 0\}}\\
    &\le \theta + \Ex{X_1,\ldots,X_n\sim\D}{\sum_{i=1}^{n}\max\{X_i - \theta, 0\}}
    =   \theta + n\cdot\Ex{X\sim\D}{\max\{X-\theta, 0\}},
    \end{align*}
    where the second step holds since $\max\{\Xmax - \theta, 0\}$ is equal to the maximum of $\max\{X_i - \theta, 0\}$ over all $i \in [n]$, which is in turn upper bounded by their sum.
\end{proof}

Another simple fact is that passing a test at $\tau_{\alpha}$ gives an expected remaining utility of $\Omega(\tau_{\alpha/2} - \tau_{\alpha})$.

\begin{lemma}\label{lemma:cond-exp}
    For any distribution $\D$ and $\alpha \in (0, 1]$,
    \[
        \Ex{X \sim \D}{X - \tau_\alpha|X > \tau_\alpha} \ge \frac{\tau_{\alpha/2} - \tau_{\alpha}}{2}.
    \]
\end{lemma}
\begin{proof}
We have
\begin{align*}
    \Ex{X\sim\D}{X-\tau_\alpha|X > \tau_\alpha}
&=   \frac{\Ex{X\sim\D}{(X-\tau_\alpha)\cdot\1{X > \tau_\alpha}}}{\pr{X\sim\D}{X > \tau_\alpha}}\\
&\ge \frac{\Ex{X\sim\D}{(\tau_{\alpha/2}-\tau_{\alpha})\cdot\1{X > \tau_{\alpha/2}}}}{\pr{X\sim\D}{X > \tau_\alpha}}\\
&=  (\tau_{\alpha/2}-\tau_{\alpha})\cdot\frac{\pr{X\sim\D}{X > \tau_{\alpha/2}}}{\pr{X\sim\D}{X > \tau_\alpha}}
=   \frac{1}{2}(\tau_{\alpha/2}-\tau_{\alpha}).
\end{align*}
The second step follows from that $(x-b)\cdot\1{x\ge b} \ge (a-b)\cdot\1{x\ge a}$ for any $x$ and $a > b$.
\end{proof}

Lemma~\ref{lemma:prophet-optimum} with $\theta = \tau_{1/n}$ shows that $\Ex{X_1, \ldots, X_n \sim \D}{\Xmax} \le \tau_{1/n} + n\cdot\Ex{X \sim \D}{\max\{X - \tau_{1/n}, 0\}}$.
In the following, we prove Theorem~\ref{thm:prophet} in the i.i.d.\ case by giving two different algorithms. The score of the first algorithm is guaranteed to match the second term, $n\cdot\Ex{X\sim\D}{\max\{X-\tau_{1/n}, 0\}}$, up to a constant factor. The second algorithm, on the other hand, achieves an expected score of $\tau_{1/n}/O(\log n)$.

\begin{proof}[Proof of Theorem~\ref{thm:prophet} (i.i.d.\ case)]
We consider a single-threshold algorithm that tests every option with the same chosen threshold $\theta$, and accepts the first option that passes the test.

\paragraph{The first algorithm.} We use threshold $\theta = \tau_{1/n}$. Since $\pr{X \sim \D}{X > \tau_{1/n}} = 1/n$, the algorithm accepts one of the $n$ options with probability $1 - (1-1/n)^n > 1-1/e$. Furthermore, conditioning on that the algorithm accepts, the expected remaining utility is
\begin{align*}
    \Ex{X\sim\D}{X-\tau_{1/n}|X > \tau_{1/n}}
&=   \frac{\Ex{X\sim\D}{(X-\tau_{1/n})\cdot\1{X > \tau_{1/n}}}}{\pr{X\sim\D}{X > \tau_{1/n}}}\\
&=   \frac{\Ex{X\sim\D}{\max\{X-\tau_{1/n}, 0\}}}{\pr{X\sim\D}{X > \tau_{1/n}}} \tag{$a\cdot\1{a > 0}=\max\{a,0\}$}\\
&=  n\cdot\Ex{X\sim\D}{\max\{X-\tau_{1/n}, 0\}}. \tag{definition of $\tau_{1/n}$}
\end{align*}

Therefore, the expected score of this algorithm matches the $n\cdot\Ex{X\sim\D}{\max\{X-\tau_{1/n}, 0\}}$ term up to a factor of $1-1/e$.

\paragraph{The second algorithm.} Let $k = \lceil \log_2 n\rceil$. Our second algorithm draws $\alpha$ uniformly at random from $\{1, 2^{-1}, 2^{-2}, \ldots, 2^{-(k-1)}\}$ and uses the threshold $\theta = \tau_\alpha$. Conditioning on the choice of $\alpha$, the probability that one of the $n$ options gets accepted is $1 - (1 - \alpha)^n > 1 - (1 - 1/n)^n > 1 - 1/e$. Furthermore, conditioning on that one of the options is accepted, the expected score is at least $\frac{\tau_{\alpha/2} - \tau_\alpha}{2}$ by Lemma~\ref{lemma:cond-exp}. Thereby, this algorithm achieves an expected score of
\[
    \frac{1}{k}\cdot(1 - 1/e)\cdot\sum_{j=0}^{k-1}\frac{\tau_{2^{-(j+1)}} - \tau_{2^{-j}}}{2}
=   \frac{(1-1/e)(\tau_{2^{-k}} - \tau_1)}{2k}
\ge   \frac{\tau_{1/n}}{O(\log n)}.
\]

An $O(\log n)$-competitive algorithm follows from randomizing between the two algorithms above.
\end{proof}

We remark that in the proof above, the final algorithm is a mixture of single-threshold algorithms. Thus, there always exists a threshold $\tau$ (that depends on $\D$) at which the single-threshold algorithm is $O(\log n)$-competitive.

\section{Prophet Setting: The General Case}\label{sec:prophet-general}
Now we tackle the general case that $\D_1, \D_2, \ldots, \D_n$ are not necessarily identical. We still assume for simplicity that each $\D_i$ is continuous. Again, Remark~\ref{remark:distribution} allows us to define $\tau^{(i)}_\alpha$ as the smallest $(1-\alpha)$-quantile of $\D_i$ for every $\alpha \in (0, 1]$. In the following, $X_1, X_2, \ldots, X_n$ are independent samples from $\D_1, \D_2, \ldots, \D_n$ respectively and we define $\Xmax \coloneqq \max_{i \in [n]}X_i$. The continuity of $\D_1, \ldots, \D_n$ implies that $\Xmax$ also follows a continuous distribution. Therefore, we may define $\tau_{\alpha}$ as the smallest $(1 - \alpha)$-quantile of $\Xmax$.

Our proof first upper bounds $\Ex{X \sim \D}{\Xmax}$ in terms of $\tau_{1/2}$ (analogously to Lemma~\ref{lemma:prophet-optimum}):
\[
    \Ex{X \sim \D}{\Xmax}
\le \tau_{1/2} + \sum_{i=1}^{n}\Ex{X_i \sim \D_i}{\max\{X_i - \tau_{1/2}, 0\}}.
\]
We will show that the single-threshold algorithm with $\theta = \tau_{1/2}$ achieves a score that is at least half of the second term above, so the main challenge is to match the first term $\tau_{1/2}$ up to an $O(\log n)$ factor. However, unlike the i.i.d.\ setting, this \emph{cannot} be done using the same threshold for every option. Instead, we will bucket the distributions based on their tails, and only test the options in the most significant group with thresholds that are chosen based on their individual distributions.

\begin{proof}[Proof of Theorem~\ref{thm:prophet}]
We start by writing the expected optimum into two parts. For each part, we will give an algorithm whose score matches the part up to an $O(\log n)$ factor. Similar to the proof of Lemma~\ref{lemma:prophet-optimum}, we have
\[
    \Ex{X \sim \D}{\Xmax}
\le   \tau_{1/2} + \Ex{X \sim \D}{\max\{\Xmax - \tau_{1/2}, 0\}}
\le \tau_{1/2} + \sum_{i=1}^n\Ex{X_i \sim \D_i}{\max\{X_i - \tau_{1/2}, 0\}}.
\]

\paragraph{Group distributions based on tails.} Let $\alpha_i = \pr{X_i \sim \D_i}{X_i > \tau_{1/2}}$ be the probability that $X_i$ exceeds the median of $\Xmax$. The definition of $\tau_{1/2}$ implies that
\[
    1/2
=   \pr{X \sim \D}{\Xmax \le \tau_{1/2}}
=   \prod_{i=1}^{n}\pr{X_i \sim \D_i}{X_i \le \tau_{1/2}}
=   \prod_{i=1}^{n}(1 - \alpha_i)
\ge 1 - \sum_{i=1}^{n}\alpha_i,
\]
and thus, $\sum_{i=1}^{n}\alpha_i \ge 1/2$. Let $k = \lceil\log_2 n\rceil$. We partition the $n$ distributions into $k + 3$ groups depending on $\alpha_i$: For $j = 0, 1, 2, \ldots, k+1$, we define
\[
    G_j \coloneqq \left\{i \in [n]: 2^{-(j+1)} < \alpha_i \le 2^{-j}\right\}.
\]
Furthermore, define $G_{k+2} \coloneqq \left\{i \in [n]: \alpha_i \le 2^{-(k+2)}\right\}$. Then, we have
\[
    1/2
\le \sum_{i=1}^{n}\alpha_i
=   \sum_{j=0}^{k+1}\sum_{i\in G_j}\alpha_i + \sum_{i \in G_{k+2}}\alpha_i
\le \sum_{j=0}^{k+1}\frac{|G_j|}{2^j} + \frac{|G_{k+2}|}{4n}
\le \sum_{j=0}^{k+1}\frac{|G_j|}{2^j} + \frac{1}{4},
\]
and thus, $\sum_{j=0}^{k+1}|G_j|/2^j \ge 1/4$. This implies that for some $j^* \in \{0, 1, \ldots, k+1\}$,
\[
    |G_{j^*}|/2^{j^*} \ge \frac{1}{4(k+2)} = \frac{1}{O(\log n)}.
\]

\paragraph{Match the $\tau_{1/2}$ term.}

We use the following algorithm:
\begin{itemize}
    \item Draw $\alpha$ randomly from some distribution over $\{1, 2^{-1}, 2^{-2}, \ldots, 2^{-j^*}\}$ to be determined later.
    \item Partition $G_{j^*}$ into $t \coloneqq \left\lceil\frac{|G_{j^*}|}{1/\alpha}\right\rceil$ blocks $B_1, B_2, \ldots, B_t$ such that: (1) Each block except $B_t$ is of size $1/\alpha$; (2) the blocks are sorted chronologically, i.e., $\max B_k < \min B_{k+1}$ for every $k \in [t - 1]$.
    \item Pick one of the $t$ blocks, $B_k$, uniformly at random.
    \item At each step $i$, ignore and reject option $i$ if $i \notin B_k$. Otherwise, test it with $\theta_i = \tau^{(i)}_\alpha$ and accept if the test passes.
\end{itemize}
We first analyze the above algorithm for fixed $\alpha \in \{1, 2^{-1}, \ldots, 2^{-j^*}\}$, and then specify the distribution from which $\alpha$ is drawn. Suppose that $i \in G_{j^*}$ is the $l$-th smallest number in its block $B_k$. Note that $l \in [1/\alpha]$. Then, we accept option $i$ if the following three happen simultaneously: (1) $B_k$ is chosen at the third step of the algorithm; (2) None of the $(l-1)$ options prior to $i$ passes the test; (3) The test passes at step $i$. All these three happen with probability
\[
    \frac{1}{t}\cdot(1 - \alpha)^{l - 1}\cdot\alpha
\ge \frac{1}{\frac{|G_{j^*}|}{1/\alpha} + 1}(1 - \alpha)^{1/\alpha-1}\cdot\alpha
\ge \frac{\alpha}{e(\alpha|G_{j^*}| + 1)}.
\]
Furthermore, by Lemma~\ref{lemma:cond-exp}, conditioning on that option $i$ passes the test at $\theta_i = \tau^{(i)}_{\alpha}$, the expected score is at least $\frac{\tau^{(i)}_{\alpha/2} - \tau^{(i)}_{\alpha}}{2}$.

Therefore, conditioning on the choice of $\alpha$, the expected score is lower bounded by
\[
    \frac{\alpha}{2e(\alpha|G_{j^*}| + 1)} \cdot \sum_{i\in G_{j^*}}(\tau^{(i)}_{\alpha/2} - \tau^{(i)}_{\alpha})
\ge \frac{\min\left\{\alpha|G_{j^*}|, 1\right\}}{4e}\cdot \frac{1}{|G_{j^*}|}\sum_{i\in G_{j^*}}(\tau^{(i)}_{\alpha/2} - \tau^{(i)}_{\alpha}).
\]
Let $Z \coloneqq \sum_{j=0}^{j^*}\frac{1}{\min\{2^{-j}|G_{j^*}|, 1\}}$. It follows from $2^{-j^*}|G_{j^*}| \ge \frac{1}{O(\log n)}$ that $Z = O(\log n)$. Then, if we set $\alpha$ to $2^{-j}$ with probability $\frac{1/\min\{2^{-j}|G_{j^*}|, 1\}}{Z}$ for each $j \in \{0, 1, \ldots, j^*\}$, our expected score is at least
\[
    \frac{1}{4eZ}\cdot\frac{1}{|G_{j^*}|}\sum_{j=0}^{j^*}\sum_{i\in G_{j^*}}(\tau^{(i)}_{2^{-(j+1)}} - \tau^{(i)}_{2^{-j}})
=   \frac{1}{4eZ|G_{j^*}|}\sum_{i\in G_{j^*}}\tau^{(i)}_{2^{-(j^*+1)}}
>   \frac{\tau_{1/2}}{4eZ}
=   \frac{\tau_{1/2}}{O(\log n)}.
\]
The second step above holds since every $i \in G_{j^*}$ satisfies $\pr{X_i \sim \D_i}{X_i > \tau_{1/2}} > 2^{-(j^*+1)}$, which implies $\tau^{(i)}_{2^{-(j^*+1)}} > \tau_{1/2}$. 

\paragraph{Match the second term.} It remains to give another algorithm with an expected score comparable to $\sum_{i=1}^n\Ex{X_i \sim \D_i}{\max\{X_i - \tau_{1/2}, 0\}}$. In fact, the single-threshold algorithm with $\tau_{1/2}$ would suffice: it gives an expected score of
\begin{align*}
    \sum_{i=1}^{n}\left[\left(\prod_{j=1}^{i-1}\pr{X_j \sim \D_j}{X_j \le \tau_{1/2}}\right)\cdot\pr{X_i \sim \D_i}{X_i > \tau_{1/2}}\cdot\Ex{X_i \sim \D_i}{X_i - \tau_{1/2}|X_i > \tau_{1/2}}\right].
\end{align*}
For each $i \in [n]$, the first term $\prod_{j=1}^{i-1}\pr{X_j \sim \D_j}{X_j \le \tau_{1/2}}$ is lower bounded by $\pr{X \sim \D}{\Xmax \le \tau_{1/2}} = 1/2$. Furthermore, $\pr{X_i \sim \D_i}{X_i > \tau_{1/2}}\cdot\Ex{X_i \sim \D_i}{X_i - \tau_{1/2}|X_i > \tau_{1/2}}$ is equal to
\[
    \Ex{X_i \sim \D_i}{(X_i - \tau_{1/2})\cdot\1{X_i > \tau_{1/2}}}
=   \Ex{X_i \sim \D_i}{\max\{X_i - \tau_{1/2}, 0\}}.
\]
Therefore, the expected score is lower bounded by $\frac{1}{2}\sum_{i=1}^n\Ex{X_i \sim \D_i}{\max\{X_i - \tau_{1/2}, 0\}}$.

Finally, randomizing between the two algorithms proves the $O(\log n)$ competitive ratio.
\end{proof}

\section{Algorithms for the Secretary Setting}\label{sec:secretary-upper}
In the secretary setting, the $n$ values $X_1, X_2, \ldots, X_n$ are guaranteed to be obtained from re-ordering (either randomly or arbitrarily) $n$ numbers $a_1, a_2, \ldots, a_n \ge 0$. Furthermore, we know either $(a_i)_{i=1}^n$ completely (full information), or only the maximum $\max_{i\in[n]}a_i$ (optimum information), or nothing at all (no information). In this section, we develop several $O(\log n)$-competitive algorithms under various combinations of the order assumption and the information about $(a_i)_{i=1}^n$.

Our starting point is the simplest case---random ordering and full information---which admits a single-threshold algorithm. We will show that a similar algorithm also succeeds when only the maximum is known. To handle the two settings in Theorem~\ref{thm:secretary-upper}, however, we need to drop either the knowledge about the maximum $a_i$, or the assumption on the arrival order. The former case can be handled by a standard technique, whereas the latter requires us to sample uniformly from a sequence of unknown length, a problem that turns out to be non-trivial. We solve this latter challenge using a more careful sampling scheme, and thus prove the theorem.

Finally, we go back to the simplest setting, and give a slightly better algorithm that is $O\left(\frac{\log n}{\log\log n}\right)$-competitive.

\subsection{Warmup: Random Order, Full or Optimum Information}
Suppose that we know $a_1, a_2, \ldots, a_n$, and $(X_i)_{i=1}^n$ is a uniformly random permutation of $(a_i)_{i=1}^n$. Let $a_{[i]}$ denote the $i$-th largest value among $a_1, \ldots, a_n$. In the following, we give a simple algorithm with an expected score of $\Omega\left(\frac{a_{[1]}}{\log n}\right)$ and thus an $O(\log n)$ competitive ratio.

Pick $k = \lfloor\log_2n\rfloor + 2 = O(\log n)$ such that $2^{k-1} > n$. For each $j \in [k]$, let
\[
    n_j \coloneqq
\left|\left\{i\in[n]: \frac{j - 1}{k}a_{[1]} < a_i \le \frac{j}{k}a_{[1]}\right\}\right|
\]
be the number of values that are approximately a $\frac{j}{k}$ fraction of the maximum. Clearly, $\sum_{j=1}^{k}n_j \le n$ and $n_k \ge 1$. We use the shorthand notation $n_{\ge j} \coloneqq n_j + n_{j+1} + \cdots + n_k$. We claim that there exists some $j^* \in \{1, 2, \ldots, k - 1\}$ such that $n_{j^*} < n_{\ge j^*+1}$. Otherwise, a simple induction shows $n_{\ge j} \ge 2^{k-j}$ for every $1 \le j \le k$, which implies $n \ge \sum_{j=1}^{k}n_j \ge 2^{k - 1} > n$, a contradiction. Furthermore, given $a_1, a_2, \ldots, a_n$, we can easily identify such an index $j^*$.

Then, the single-threshold algorithm at $\theta = \frac{j^* - 1}{k}a_{[1]}$ succeeds. Recall that this algorithm accepts the first $X_i$ that exceeds $\theta$ and gives a score of $X_i - \theta$. By definition, the number of options that could pass the test is exactly $n_{\ge j^*}$. Since $n_{j^*} < n_{\ge j^*+1} = n_{\ge j^*} - n_{j^*}$, at least half of these $n_{\ge j^*}$ options (initially) have values $\ge\frac{j^*}{k}a_{[1]}$. Given that the values are shuffled uniformly at random, with probability at least $1/2$, the remaining utility of the option that we accept is at least
\[
    \frac{j^*}{k}a_{[1]} - \theta
=   \frac{j^*}{k}a_{[1]} - \frac{j^* - 1}{k}a_{[1]}
=   \frac{a_{[1]}}{k}.
\]
Our expected utility is then lower bounded by $\frac{1}{2}\cdot\frac{a_{[1]}}{k} = \Omega\left(\frac{a_{[1]}}{\log n}\right)$.

\paragraph{Extension to optimum information case.} In the algorithm above, the knowledge of $(a_i)_{i=1}^n$ is only used for choosing the right $j^*$ in the threshold. We show that a similar algorithm that picks $j^*$ randomly is still $O(\log n)$ competitive, and thus can be applied to the optimum information case.

Under the same definition of $k$, $n_j$ and $n_{\ge j}$ as above, we draw $j$ uniformly at random from $[k-1]$, and run the single-threshold algorithm with $\theta = \frac{j-1}{k}\cdot a_{[1]}$. Conditioning on the choice of $j$, we accept one of the $n_{\ge j}$ options with value $> \theta$ uniformly at random, so our utility is at least $\frac{a_{[1]}}{k}$ with probability at least $\frac{n_{\ge j+1}}{n_{\ge j}}$. Averaging over the randomness in $j$ lower bounds the expected score by
\[
    \frac{a_{[1]}}{k}\cdot\frac{1}{k-1}\sum_{j=1}^{k-1}\frac{n_{\ge j+1}}{n_{\ge j}}
\ge \frac{a_{[1]}}{k}\cdot\left(\prod_{j=1}^{k-1}\frac{n_{\ge j+1}}{n_{\ge j}}\right)^{\frac{1}{k-1}}
=   \frac{a_{[1]}}{k}\cdot\left(\frac{n_{\ge k}}{n_{\ge 1}}\right)^{\frac{1}{k-1}}
\ge \frac{a_{[1]}}{k}\cdot n^{-\frac{1}{k-1}}.
\]
Our choice of $k$ gives $n^{-\frac{1}{k-1}} > 1/2$, so our score is a $\frac{1}{2k} = \frac{1}{O(\log n)}$ fraction of the optimum $a_{[1]}$.

\subsection{Random Order, with No Information}
We further discard the knowledge of $a_{[1]}$ and prove Case~(1) of Theorem~\ref{thm:secretary-upper}: unknown values $a_1, a_2, \ldots, a_n$ arrive in a uniformly random order. A simple idea that is often used in secretary problems is to observe half of the options, and use the largest value among them as an estimate of $a_{[1]}$. If $a_{[1]}$ appears in the latter half, while this estimate is within a constant factor to $a_{[1]}$, we would obtain a competitive algorithm by making the argument in the previous section robust. However, this never holds if $a_{[1]}$ is much larger than the second largest value $a_{[2]}$. Fortunately, another simple algorithm works for this case---single-threshold with the maximum among the first half as the threshold.

To formalize this idea, we state the following lemma that further generalize the ``warmup'' algorithm to the case where we are only given a ``hint'' about $a_{[1]}$.

\begin{lemma}\label{lemma:secretary-random-approx-opt}
    In the secretary setting under random order, there is an algorithm that, given any $\Xmaxhat$ that lies in $[0, a_{[1]}]$, achieves a score of at least $\frac{\Xmaxhat}{O(\log n)}$ in expectation.
\end{lemma}

The proof of the lemma follows from simply replacing $a_{[1]}$ with $\Xmaxhat$ in the argument above; we give a formal proof in Appendix~\ref{sec:omitted-secretary-upper} for completeness.

\begin{proof}[Proof of Theorem~\ref{thm:secretary-upper}, Case (1)]
    We observe the values of the first $m \coloneqq \lfloor n/2\rfloor$ options (by using threshold $+\infty$), and let $\Xmaxhat$ be the largest among them. With probability $\frac{m\cdot(n - m)}{n \cdot (n - 1)} \ge 1/4$, an option with value $a_{[2]}$ appears in the first $m$ ones, while an option with value $a_{[1]}$ is in the remaining $n-m$ ones. This implies $\Xmaxhat = a_{[2]} \le a_{[1]}$ and we condition on this event in the following.
    
    Conditioning on the options that appear in the first $m$ steps, the remaining $n - m$ options are still a uniformly random permutation of the unseen ones. Therefore, applying the algorithm from Lemma~\ref{lemma:secretary-random-approx-opt} with the hint $\Xmaxhat$ gives an expected score of at least $\frac{\Xmaxhat}{O(\log n)} = \frac{a_{[2]}}{O(\log n)}$. On the other hand, if we simply test each of the last $n - m$ options with threshold $\Xmaxhat$ and accept the first one that passes the test, we will either accept an option with initial value $a_{[1]}$ (if $a_{[1]} > a_{[2]}$), or accept nothing (if $a_{[1]} = a_{[2]}$). In either case, our score is $a_{[1]} - a_{[2]}$. Therefore, using one of the two strategies randomly gives an expected score of at least
    \[
        \frac{1}{4}\cdot\frac{1}{2}\left(\frac{a_{[2]}}{O(\log n)} + a_{[1]} - a_{[2]}\right)
    =   \frac{a_{[1]}}{O(\log n)}.
    \]
\end{proof}

\subsection{Arbitrary Order, with Optimum Information}\label{sec:bit-sampling}
We give another $O(\log n)$-competitive algorithm when $a_1, a_2, \ldots, a_n$ arrive in an arbitrary order, and we only know the maximum value $a_{[1]}$. Recall that the ``warmup'' algorithm is the single-threshold algorithm with a randomly chosen $\theta = \frac{j - 1}{k}a_{[1]}$, where $k = \Theta(\log n)$ and $j$ is uniformly drawn from $[k - 1]$. The analysis uses the following two observations. First, among the options that could pass the test, a constant fraction of them are ``good'' in the sense that their values are higher than $\theta + a_{[1]}/k$, which give a score of $\ge a_{[1]}/k$. Second, since the options arrive in a random order, the one that we accept is ``good'' with $\Omega(1)$ probability.

The same algorithm fails in the arbitrary order case, since the options might be adversarially ordered such that the ``good'' options always appear after the ``bad'' ones. If we (hypothetically) knew the number $n_{\ge j}$ of total options that could pass the test at $\theta$, this would not pose a challenge, since we could instead draw $l$ uniformly at random from $[n_{\ge j}]$ and accept the $l$-th option that passes the test. Unfortunately, we cannot obtain (or even estimate) $n_{\ge j}$ without knowing $a_1, a_2, \ldots, a_n$.

\paragraph{A bit sampling game.} The following ``bit sampling'' game is an abstraction of this challenge. An adversary picks a binary sequence of an unknown length $m \in [n]$, with the only restriction that the fraction of ones is strictly higher than $1/2$. The player, knowing $n$ but not $m$, observes the bits one by one, and may choose to commit to the next unseen bit at any point (including before seeing any bits). The player wins if the chosen bit is a ``$1$'', and loses if it either commits to a ``$0$'', or fails to select a bit before the end of the sequence.

The player would easily win with probability $> 1/2$ if it could select one of the $m$ bits uniformly at random. This uniform sampling would be possible if either the sequence length $m$ were known, or the bits came in a random order. When the order is arbitrary and $m$ is unknown, it \emph{might} appear that the player's only strategy is to guess the sequence length $\hat m$, and to sample one of the first $\hat m$ bits uniformly. Unfortunately, this succeeds only if $\hat m$ is within a constant factor to the actual sequence length $m$, which at best happens with probability $O\left(\frac{1}{\log n}\right)$.

Perhaps surprisingly, with a better strategy, the player wins the game with a \emph{constant} probability regardless of the maximum sequence length $n$.

\begin{lemma}\label{lemma:bit-sampling}
    In the bit sampling game, the player has a strategy that wins the game with probability at least $\frac{1}{6}$.
\end{lemma}

In the following, we first show how this result implies an online pen testing algorithm (with arbitrary order, optimum information), and then prove Lemma~\ref{lemma:bit-sampling}. We will actually prove a slightly more general result, which immediately implies Case~(2) of Theorem~\ref{thm:secretary-upper} and will be useful in later sections.

\begin{lemma}[Strengthening of Theorem~\ref{thm:secretary-upper}, Case~(2)]\label{lemma:secretary-arbitrary-approx-opt}
    In the secretary setting under arbitrary order, there is an algorithm that, given any $\Xmaxhat$ that lies in $[0, a_{[1]}]$, achieves a score of at least $\frac{\Xmaxhat}{O(\log n)}$ in expectation.
\end{lemma}

\begin{proof}[Proof of Lemma~\ref{lemma:secretary-arbitrary-approx-opt}]
Let $k = 2\lfloor\log_2n\rfloor + 2$, draw $j^*$ uniformly at random from $[k - 1]$, and set $\theta \coloneqq \frac{j^*-1}{k}\cdot \Xmaxhat$. For each $j \in [k]$, let $n_{\ge j} \coloneqq |\{i \in [n]: a_i > \frac{j - 1}{k}\cdot \Xmaxhat\}|$. Note that $n_{\ge 1} \le n$ and $n_{\ge k} \ge 1$.

We claim that with probability $\ge 1/2$, $j^*$ satisfies that $n_{\ge j^*} < 2n_{\ge j^*+1}$. Suppose otherwise, that there exist $k/2$ different values $1 \le j_1 < \cdots < j_{k/2} \le k - 1$ such that $n_{\ge j} \ge 2n_{\ge j + 1}$ holds for every $j \in \{j_1, \ldots, j_{k/2}\}$. This would give
\[n \ge n_{\ge 1} \ge 2^{k/2}n_{\ge k} \ge 2^{\lfloor\log_2n\rfloor + 1}\cdot 1 > n,\]
a contradiction. We condition on the event that $n_{\ge j^*} < 2n_{\ge j^* + 1}$ in the following.

Since $n_{\ge j^* + 1} > \frac{n_{\ge j^*}}{2}$, among the $n_{\ge j^*}$ options that could pass the test at $\theta$, strictly more than half of them would leave a remaining utility of at least $\frac{\Xmaxhat}{k}$. We call them the ``good'' options, and the other options (that pass the test but leave a utility $<\frac{\Xmaxhat}{k}$) the ``bad'' ones.

Now we simulate the player from Lemma~\ref{lemma:bit-sampling} on a hypothetical bit sampling instance with maximum sequence length $n$. 
We test the options one by one with the same threshold $\theta$. If an option $i$ passes the test, we check whether the player in the bit sampling problem commits to the next bit. If so, we accept option $i$; otherwise, we further test the option to check whether its initial value is above $\frac{j^* + 1}{k}\cdot \Xmaxhat$. If so, we feed a ``$1$'' to the bit sampling player and feed a ``$0$'' otherwise.

Assuming that $n_{\ge j^*} < 2n_{\ge j^* + 1}$ holds, the probability that we end up with a score $\ge \frac{\Xmaxhat}{k}$ is lower bounded by the player's winning probability, which is $\ge 1/6$ by Lemma~\ref{lemma:bit-sampling}. This lower bounds our expected score by $\frac{1}{2}\cdot\frac{1}{6}\cdot\frac{\Xmaxhat}{k}
=   \frac{\Xmaxhat}{O(\log n)}$ and proves the lemma.
\end{proof}

\paragraph{A winning strategy for bit sampling.} Our strategy for the bit sampling game crucially keeps track of the difference, denoted by $\Delta$, between the number of zeros and ones. Intuitively, when $\Delta = 0$, the adversary might be tempted to set the next bit to ``$1$'' and end the sequence. Thus, the player must commit with $\Omega(1)$ probability in order to ``catch'' this bit. On the other hand, the player can be less aggressive when $\Delta$ is large.

\begin{proof}[Proof of Lemma~\ref{lemma:bit-sampling}]
Consider the following strategy of the player:
\begin{itemize}
    \item Before seeing each bit, let $\Delta$ denote the number of zeros minus the number of ones, among all bits that have appeared so far.
    \item Commit to the next bit with probability $2^{-(\Delta+2)}$. With the remaining probability, observe the next bit and proceed.
\end{itemize}

We prove the following claim by an induction on $t$: Assuming that $n - t$ bits have been observed, and the number of zeros is higher than the number of ones by $\Delta \ge 0$, the player following the above strategy wins with probability at least $\frac{1}{3}\left(1 - 2^{-(\Delta+1)}\right)$. Applying this claim with $t = n$ and $\Delta = 0$ shows that the player wins with probability $\ge \frac{1}{6}$.

At the base case that $t = 0$, since the sequence contains more ones than zeros, $\Delta$ cannot be non-negative and there is nothing to prove. Now we proceed to the inductive step and assume that the claim holds for $t - 1$. Since currently the number of zeros is greater than or equal to the number of ones, the sequence is not over and there must be a next bit. First suppose that the next bit is $1$. With probability $2^{-(\Delta+2)}$, the player selects the next bit and wins the game. With the remaining probability, we reach the $t-1$ case where the difference between the bit counts is $\Delta - 1$. If $\Delta \ne 0$, we can apply the inductive hypothesis and lower bound the winning probability by $\frac{1}{3}\left(1 - 2^{-\Delta}\right)$; when $\Delta = 0$, the winning probability is trivially lower bounded by $0 = \frac{1}{3}\left(1 - 2^{-\Delta}\right)$. Thus, in either case, the overall winning probability is at least
\[
    2^{-(\Delta+2)} + \left(1 - 2^{-(\Delta+2)}\right)\cdot\frac{1}{3}\left(1 - 2^{-\Delta}\right)
=   \frac{1}{3}\left(1 - 2^{-(\Delta+1)} + 2^{-(2\Delta+2)}\right)
\ge \frac{1}{3}\left(1 - 2^{-(\Delta+1)}\right).
\]
Similarly, if the next bit is $0$, the algorithm wins with probability at least
\[
    \left(1 - 2^{-(\Delta+2)}\right)\cdot\frac{1}{3}\left(1 - 2^{-(\Delta+2)}\right)
\ge \frac{1}{3}\left(1 - 2^{-(\Delta+1)}\right).
\]
This completes the inductive step and proves the lemma.
\end{proof}

\subsection{An Improved Algorithm under Random Ordering, with Full Information}
We finish the section by proving the upper bound part of Theorem~\ref{thm:secretary-improved}, which shaves a $\log\log n$ factor off the competitive ratio in the random order, full information case. Note that given $a_1, a_2, \ldots, a_n$ and having observed $X_1$ through $X_i$, the player knows the $n - i$ unseen values up to a random permutation, namely, $\{X_{i+1}, X_{i+2}, \ldots, X_n\} = \{a_1, a_2, \ldots, a_n\}\setminus\{X_1, X_2, \ldots, X_i\}$.

We say that a ``gap'' appears in $\{X_{i+1}, X_{i+2}, \ldots, X_n\}$ if, for some $A < B$, the set contains at least one element larger than $B$, but nothing in interval $(A, B]$. The key observation is that, whenever such a gap appears, the player can secure a score of $B - A$ by testing the remaining options at threshold $A$ and accepting the first one that passes the test, as it must leave a remaining utility of at least $B - A$. Our proof of Theorem~\ref{thm:secretary-improved} essentially shows that with a good probability, a gap of size $\approx \frac{\log\log n}{\log n}\cdot a_{[1]}$ appears at some point $i$.

\begin{proof}[Proof of Theorem~\ref{thm:secretary-improved} (Upper Bound)]
Let $k \ge 1$ be an integer to be determined later. For each $j \in [k]$, define $n_j \coloneqq |\{i \in [n]: \frac{j - 1}{k}a_{[1]} < a_i \le \frac{j}{k}a_{[1]}\}|$ and let $n_{\ge j} \coloneqq n_j + n_{j+1} + \cdots + n_k$ be the number of options with utility $> \frac{j - 1}{k} a_{[1]}$.

Before testing each option $i$, we review the multiset of unseen values. If, for some $j \in \{1, \ldots, k-1\}$, none of the values lies in $\left(\frac{j - 1}{k}a_{[1]}, \frac{j}{k}a_{[1]}\right]$ and at least one of them is higher than $\frac{j}{k}a_{[1]}$, we test each remaining option at threshold $\frac{j-1}{k}a_{[1]}$ and accept the first one that passes the test. Otherwise, we test option $i$ at threshold $+\infty$ (so that we see $X_i$) and reject it. Note that if the condition holds for any $j$ at any point $i$, the player receives a score of at least $a_{[1]}/k$. In the remainder of the proof, we show that this is indeed the case with probability $\Omega(1)$, for some carefully chosen $k$.

Let $\E_j$ denote the event that, among the $n_{\ge j}$ options with value $> \frac{j - 1}{k}a_{[1]}$, the one that appears last has value $> \frac{j}{k}a_{[1]}$. Note that event $\E_j$ implies that our algorithm would detect this gap of size $\frac{a_{[1]}}{k}$, and thus secure a score $\ge\frac{a_{[1]}}{k}$. Given that the options are randomly ordered, $\E_j$ happens with with probability exatly $\frac{n_{\ge j + 1}}{n_{\ge j}}$. Furthermore, it can be verified that $\E_1, \E_2, \ldots, \E_{k-1}$ are independent. Thus, the probability that none of the events happens is
\[
    \prod_{j=1}^{k-1}\left(1 - \frac{n_{\ge j+1}}{n_{\ge j}}\right)
=   \exp\left((k-1)\cdot\frac{1}{k-1}\sum_{j=1}^{k-1}\ln\left(1 - \frac{n_{\ge j+1}}{n_{\ge j}}\right)\right).
\]
Note that $\prod_{j=1}^{k-1}\frac{n_{\ge j+1}}{n_{\ge j}} = \frac{n_{\ge k}}{n_{\ge 1}} \ge \frac{1}{n}$. By the AM-GM inequality,
\[
    \frac{1}{k-1}\sum_{j=1}^{k-1}\frac{n_{\ge j+1}}{n_{\ge j}}
\ge \left(\prod_{j=1}^{k-1}\frac{n_{\ge j+1}}{n_{\ge j}}\right)^{\frac{1}{k-1}}
\ge n^{-1/(k-1)}.
\]
Then, using $\ln(1-x)\le -x$, we have
\[
    \frac{1}{k-1}\sum_{j=1}^{k-1}\ln\left(1 - \frac{n_{\ge j+1}}{n_{\ge j}}\right)
\le -\frac{1}{k-1}\sum_{j=1}^{k-1}\frac{n_{\ge j+1}}{n_{\ge j}}
\le -n^{-1/(k-1)},
\]
and the probability that none of $\E_1, \E_2, \ldots, \E_{k-1}$ happens is at most $\exp\left(-(k-1)\cdot n^{-1/(k-1)}\right)$.

We pick $k = \Theta\left(\frac{\log n}{\log\log n}\right)$ such that $(k-1)\cdot n^{-1/(k-1)} \ge 1$, which guarantees that the algorithm ends up with score $\ge \frac{a_{[1]}}{k}$ with probability $\ge 1 - e^{-1}$. The expected utility is thereby at least $(1 - e^{-1})\cdot\frac{a_{[1]}}{k} = \Omega\left(\frac{\log\log n}{\log n}\cdot a_{[1]}\right)$. In other words, the algorithm is $O\left(\frac{\log n}{\log\log n}\right)$-competitive.
\end{proof}

\section{Lower Bounds for the Secretary Setting}\label{sec:secretary-lower}
In this section, we prove all the lower bounds for the secretary setting. We start with the random order, full information case in Theorem~\ref{thm:secretary-improved}. We then prove the second case of Theorem~\ref{thm:secretary-lower} (arbitrary order, full information) using a similar sequence $(a_i)_{i=1}^n$, but under a harder distribution over permutations of the values. Finally, we prove Case~(1) of Theorem~\ref{thm:secretary-lower}, essentially by reducing the setting to the i.i.d.\ case of the prophet setting, which has a simple lower bound construction based on the exponential distribution (Fact~\ref{fact:lower}).

\subsection{Random Order, Full Information}
Fix integer $k \ge 1$ and let $n = 2^{k+1} - 1$. We consider the sequence $a_1, a_2, \ldots, a_n$ in which each $j \in \{0, 1, \ldots, k\}$ appears $2^{k-j}$ times. We will show that, when $a_1, a_2, \ldots, a_n$ arrive in a uniformly random order, the highest expected score that the player can achieve is $O(\log k)$. This would establish the $\frac{k}{O(\log k)} = \Omega\left(\frac{\log n}{\log\log n}\right)$ lower bound on the competitive ratio.

We first note that the optimal algorithm for this specific instance should satisfy a few constraints.
\begin{remark}\label{remark:opt-alg}
Since all the values are integers between $0$ and $k$ in this problem instance, we can assume without loss of generality that each threshold $\theta_i$ chosen by the player is among $\{0, 1, \ldots, k\}$. Indeed, picking a threshold higher than $k$ is equivalent to picking $\theta_i = k$, and a non-integral threshold $\theta_i$ gives the same information as $\lfloor\theta_i\rfloor$, but leaves a smaller remaining value. Under this assumption, the score of the player always lies in $\{0, 1, \ldots, k\}$.
\end{remark}

The following lemma states that, for any $\Delta \in \{1, 2, \ldots, k\}$, the probability that the player gets a score $\ge \Delta$ is exponentially small in $\Delta$.

\begin{lemma}\label{lemma:secretary-improved-lower}
    In the instance defined as above, for any $\Delta \in [k]$ and $\theta \in \{0, 1, \ldots, k - \Delta\}$, the probability that the player gets a score of $\ge \Delta$ after accepting an option that has been tested at $\theta_i = \theta$ is at most $4\cdot 2^{-\Delta}$.
\end{lemma}

We first show, via a straightforward calculation, that the lower bound in Theorem~\ref{thm:secretary-improved} immediately follows from the lemma.

\begin{proof}[Proof of Theorem~\ref{thm:secretary-improved} (Lower Bound)]
For each $\Delta \in [k]$, Lemma~\ref{lemma:secretary-improved-lower} together with a union bound over $\theta \in \{0, 1, \ldots, k - \Delta\}$ shows that the probability of achieving score $\ge \Delta$ is at most $(k - \Delta + 1)\cdot 4\cdot 2^{-\Delta} \le 4k \cdot 2^{-\Delta}$.
Therefore, the expected score of the player is upper bounded by
\[
    \sum_{\Delta=1}^{k}\pr{}{\text{score} \ge \Delta}
\le \sum_{\Delta=1}^{+\infty}\min\{1, 4k\cdot 2^{-\Delta}\}
=   O(\log k).
\]
Since $a_{[1]} = k$, the competitive ratio is at least $\frac{k}{O(\log k)} = \Omega\left(\frac{\log n}{\log\log n}\right)$.
\end{proof}

Now we prove Lemma~\ref{lemma:secretary-improved-lower}.

\begin{proof}[Proof of Lemma~\ref{lemma:secretary-improved-lower}]
    Fix $\Delta \in [k]$ and an integer $\theta$ between $0$ and $k - \Delta$. We say that the player ``wins'', if it gets a score of at least $\Delta$ after testing the accepted option with threshold $\theta$. We may assume without loss of generality that the player, to maximize its probability of winning, uses either of the following two strategies at each option $i$:
    \begin{itemize}
        \item \textbf{(Observe)} The player observes $X_i$ by testing the option at $\theta = k$ and then rejecting it.
        \item \textbf{(Commit)} The player tests $X_i$ at $\theta_i = \theta$, and accepts if the test passes.
    \end{itemize}
    Indeed, using any threshold $\theta_i \ne \theta$ does not count towards the winning probability; the player may as well observe $X_i$ perfectly by picking $\theta_i = k$.
    
    Among the options that could pass a test at $\theta$, we say that an option is ``good'' if its value is at least $\theta + \Delta$, and ``bad'' otherwise. There are $G \coloneqq \sum_{j=\theta + \Delta}^{k}2^{k-j}$ good options and $B \coloneqq \sum_{j=\theta+1}^{\theta+\Delta-1}2^{k-j}$ bad ones.
    
    We will prove by induction on $t$ that, when there are exactly $t$ remaining options and they contain $b$ bad ones and $g$ good ones ($b + g \le t$), the probability of accepting a good option is at most $\frac{g}{b + g}$ (interpreted as $0$ if $b = g = 0$).
    
    At the base case $t = 0$, we can only have $b = g = 0$ and the claim clearly holds. Suppose that the claim holds at $t' = t - 1$, and we fix $b, g$ such that $b + g \le t$. First assume that the player chooses to \emph{commit}. With probability $1 - \frac{b + g}{t}$, the next option fails the test at $\theta$, and the player moves on. The conditional winning probability is still $\le \frac{g}{b + g}$ by the inductive hypothesis. With probability $\frac{g}{t}$, the next option is good and is accepted by the player, in which case the player wins. Similarly, the player accepts a bad option and loses with probability $\frac{b}{t}$. Thus, the overall winning probability is upper bounded by
    \[
        \left(1 - \frac{b + g}{t}\right)\cdot\frac{g}{b + g} + \frac{g}{t}\cdot 1 + \frac{b}{t} \cdot 0
    =   \frac{g}{b + g}.
    \]
    
    If the player decides to \emph{observe}, a similar reasoning shows that the player wins with probability at most
    \[
        \left(1 - \frac{b + g}{t}\right)\cdot\frac{g}{b + g} + \frac{g}{t}\cdot \frac{g - 1}{(g - 1) + b} + \frac{b}{t} \cdot \frac{g}{g + (b - 1)}
    =   \frac{g}{b + g}.
    \]
    
    Therefore, applying the claim at $t = n$, $b = B$ and $g = G$ shows that the probability in question is at most
    \[
        \frac{G}{B + G}
    =   \frac{\sum_{j=\theta + \Delta}^{k}2^{k-j}}{\sum_{j=\theta + 1}^{k}2^{k-j}}
    \le \frac{2\cdot 2^{k-(\theta + \Delta)}}{2^{k-(\theta + 1)}}
    =   4\cdot2^{-\Delta}.
    \]
\end{proof}

\subsection{Arbitrary Order, Full Information}\label{sec:non-uniform-order}
Now we move on to the worst-case order setting, and try to strengthen the previous lower bound to $\Omega(\log n)$. We will actually construct a hard distribution over permutations of some fixed sequence. In light of the $O\left(\frac{\log n}{\log\log n}\right)$ upper bound in Theorem~\ref{thm:secretary-improved}, this hard distribution cannot be uniform.

Let $k \ge 1$ be an integer. We generate an instance with $n = 4^0 + 4^1 + \cdots + 4^k$ options in total by repeating the following two steps:
\begin{itemize}
    \item Sample $X$ from the geometric distribution that takes each value $j \in \{0, 1, 2, \ldots\}$ with probability $2^{-(j+1)}$.
    \item If $X \le k$ and $X$ has appeared less than $4^{k-X}$ times in the sequence, we append value $X$ to the sequence.
\end{itemize}
Clearly, the resulting sequence is always a permutation of the length-$n$ sequence in which each $j \in \{0, 1, \ldots, k\}$ appears exactly $4^{k-j}$ times.

By the same argument as in Remark~\ref{remark:opt-alg}, we may assume that player picks each threshold $\theta_i$ from $\{0, 1, \ldots, k\}$, so that the resulting score is always among $\{0, 1, \ldots, k\}$. The key step of the lower bound is the following lemma.

\begin{lemma}\label{lemma:secretary-lower}
    In the instance defined as above, for any integer $\Delta \ge 3$, the probability of getting a score of exactly $\Delta$ is at most $c\cdot 2^{-\Delta}$, where $c = 21$ is a universal constant.
\end{lemma}

Compared to Lemma~\ref{lemma:secretary-improved-lower}, the lemma above is stronger in that it bounds the total probability of getting a high score, even after a union bound over all the possible thresholds $\theta$, by an exponentially small quantity in $\Delta$.

Assuming Lemma~\ref{lemma:secretary-lower}, the second case of Theorem~\ref{thm:secretary-lower} follows from a simple calculation:
\begin{proof}[Proof of Theorem~\ref{thm:secretary-lower}, Case (2)] The expected score of any algorithm is upper bounded by
\[
    \sum_{\Delta = 1}^{k}\Delta\cdot\pr{}{\textrm{score} = \Delta}
\le 3 + c\cdot\sum_{\Delta = 3}^{+\infty}\Delta\cdot 2^{-\Delta}
=   O(1),
\]
where the second step follows from Lemma~\ref{lemma:secretary-lower}. This proves the $\frac{k}{O(1)} = \Omega(\log n)$ lower bound.
\end{proof}

\begin{proof}[Proof of Lemma~\ref{lemma:secretary-lower}]
Fix $\Delta \ge 3$. We say that the player \emph{wins} if its score is exactly $\Delta$, and assume that the player tries to maximize its probability of winning according to this criterion. When the player accepts an option $i$ after testing it with threshold $\theta_i$, we say that this acceptance is \emph{risky} if at least one of the unseen options has value $\theta_i + 1$; the acceptance of option $i$ is called \emph{safe} otherwise.

In the following, we will show that the probability that the player wins (i.e., obtain score exactly $\Delta$) by accepting an option riskily is at most $2^{1 - \Delta}$, while the probability of a safe acceptance is also small (in fact, doubly exponential in $\Delta$). The lemma would then follow easily.

\paragraph{Probability of a risky win.} We condition on the event that the player riskily accepts option $i$ after testing it with threshold $\theta_i$, as well as the realization of $X_1, X_2, \ldots, X_{i-1}$. For each $j \in \{0, 1, \ldots, k\}$, we say that $j$ is \emph{active}, if $j$ appears at least once in $X_i, X_{i+1}, \ldots, X_n$. Clearly, the definition of a risky acceptance implies that $\theta_i + 1$ is active.

Our construction of the sequence implies that, conditioning on $X_1, X_2, \ldots, X_{i-1}$, the next value $X_i$ is distributed over all active values $j$, with the probability of $X_i = j$ proportional to $2^{-j}$.
Therefore, given that $\theta_i + 1$ is active, the conditional probability of the event $X_i = \theta_i + \Delta$ is upper bounded by $\frac{2^{-(\theta_i + \Delta)}}{2^{-(\theta_i + 1)} + 2^{-(\theta_i + \Delta)}} \le 2^{1-\Delta}$. This means that conditioning on a risky acceptance, the probability of winning the game (i.e., exactly getting score $\Delta$) is at most $2^{1 - \Delta}$. Consequently, the probability of a risky win is at most $2^{1 - \Delta}$.

\paragraph{Probability of a safe win.}
For $\theta \in \{0, 1, \ldots, k - \Delta\}$, we call the options with value $\theta + 1$ the ``$\theta$-bad'' ones, and those with value $\theta + \Delta$ the ``$\theta$-good'' ones. Let $\E_\theta$ be the event that the last $\theta$-good option appears after the last $\theta$-bad option. For the player to get score $\Delta$ by safely accepting option $i$, it must be the case that $X_i = \theta_i + \Delta$, while $\theta_i + 1$ never appears in $X_{i+1}, X_{i+2}, \ldots, X_n$. In other words, event $\E_{\theta_i}$ must happen. Therefore, we can upper bound the probability of a safe win by controlling the probability of each $\E_\theta$.

Let $G \coloneqq 4^{k - (\theta + \Delta)}$ and $B \coloneqq 4^{k - (\theta + 1)}$ denote the total numbers of $\theta$-good and -bad options. We focus on the length-$(B+G)$ subsequence of $X_1, X_2, \ldots, X_n$ consisting of only $\theta + 1$ and $\theta + \Delta$. Our construction of $(X_i)_{i=1}^n$ guarantees that this subsequence follows the same distribution as the output of the following procedure:
\begin{itemize}
    \item Sample $Y_1, Y_2, \ldots, Y_{B+G}$ independently from the Bernoulli distribution with mean $\mu = \frac{2^{-\Delta}}{2^{-1} + 2^{-\Delta}}$.
    \item For each $i = 1, 2, \ldots, B+G$, if $Y_i = 0$, append $\theta + 1$ to the end of the sequence, and append $\theta + \Delta$ otherwise. Repeat this step until either $\theta + 1$ has appeared $B$ times, or $\theta + \Delta$ has appeared $G$ times.
    \item In the former case, append $\theta + \Delta$ to the end of the sequence until the sequence has length $B + G$; append $\theta + 1$ in the latter case. 
\end{itemize}

Note that $\E_\theta$ is exactly the event that the last entry of the length-$(B+G)$ sequence is $\theta + \Delta$. For this to happen, we must have $Y_1 + Y_2 + \cdots + Y_{B+G} \le G$, which, by a Chernoff bound, happens with probability at most
\[
    \exp\left(-\frac{1}{2}\cdot\left(1 - \frac{G}{\mu(B + G)} \right)^2\cdot \mu(B + G)\right).
\]
Recall that $B = 4^{k-(\theta + 1)}$, $G = 4^{k-(\theta + \Delta)}$ and $\mu = \frac{2^{-\Delta}}{2^{-1} + 2^{-\Delta}}$. We have
$\frac{G}{\mu(B + G)} = \frac{2^{-1} + 2^{-\Delta}}{2^{\Delta}(4^{-1} + 4^{-\Delta})} \le \frac{1}{2}$, where the last step holds for $\Delta \ge 3$. This simplifies the bound into
\[
    \pr{}{\E_{\theta}}
\le \exp\left(-\frac{\mu (B+G)}{8}\right)
\le \exp\left(-\frac{2^{-\Delta}\cdot 4^{k-\theta-1}}{8}\right)
\le \exp\left(-\frac{2^{k-\theta}}{32}\right).
\]
The last step above holds since $k - \theta \ge \Delta$.
The probability that at least one of $\E_\theta$ happens is then upper bounded by
\[
    \sum_{\theta=0}^{k-\Delta}\pr{}{\E_\theta}
\le \sum_{\theta=0}^{k-\Delta}\exp\left(-\frac{2^{k-\theta}}{32}\right)
=   \sum_{j = \Delta}^k\exp\left(-\frac{2^j}{32}\right).
\]

In total, the winning probability of the player (either risky or safe) is at most
\[2^{1-\Delta} + \sum_{j = \Delta}^k\exp\left(-\frac{2^j}{32}\right),\]
which is at most $21\cdot 2^{-\Delta}$ for any $\Delta \ge 3$.
\end{proof}

\subsection{Random Order, Optimum Information}
Finally, we prove the first case of Theorem~\ref{thm:secretary-lower}: The competitive ratio is still $\Omega(\log n)$, even if the options arrive in a random order and we are given the maximum value. We prove the lower bound by considering a distribution over $(a_1, a_2, \ldots, a_n)$. The distribution is chosen such that the expectation of $a_{[1]}$ is at least $\Omega(\log n)$, while any pen testing algorithm, when given $a_{[1]}$ and running on a random permutation of $(a_i)_{i=1}^n$, achieves an $O(1)$ score in expectation.

Let $\D$ be the distribution of $\min\left\{X, \frac{\ln n}{2}\right\}$ when $X$ is drawn from the exponential distribution with parameter $1$. In other words, we truncate the exponential by moving all the probability mass from its tail $\left[\frac{\ln n}{2}, +\infty\right)$ to a point mass at $\frac{\ln n}{2}$. We draw $a_1, a_2, \ldots, a_n$ from $\D$ independently. Intuitively, one of the $a_1, a_2, \ldots, a_n$ would take value $\frac{\ln n}{2}$ with high probability, so the player gains little information from $a_{[1]}$. Without this additional information, the player is essentially in the i.i.d.\ prophet setting, working on the lower bound instance from Fact~\ref{fact:lower}.

\begin{proof}[Proof of Theorem~\ref{thm:secretary-lower}, Case (1)]
Let $a_1, a_2, \ldots, a_n$ be drawn independently from $\D$. We first show that w.h.p., $a_{[1]} = \frac{\ln n}{2}$. Indeed, the probability of never getting a sample $\frac{\ln n}{2}$ is given by:
\[
    \left(1 - e^{-\frac{\ln n}{2}}\right)^n
=   \left(1 - \frac{1}{\sqrt{n}}\right)^n
\le e^{-\sqrt{n}}.
\]
Therefore, the expectation of $a_{[1]}$ is at least $(1 - e^{-\sqrt{n}})\cdot\frac{\ln n}{2} = \Omega(\log n)$.

Then, we argue that no player could get an expected score strictly higher than $1 + e^{-\sqrt{n}}\cdot\frac{\ln n}{2} = O(1)$, when the value of $a_{[1]}$ is given and $(X_i)_{i=1}^n$ is chosen as a random permutation of $(a_i)_{i=1}^n$. Here, the expectation is over the choice of $(a_i)_{i=1}^n$, the random arrival order, as well as the randomness in the player itself. Suppose that this is not true. Then, we consider running the same algorithm on the same distribution over problem instances, except that the player is always given value $\frac{\ln n}{2}$ instead of $a_{[1]}$. Since $a_{[1]} \ne \frac{\ln n}{2}$ only happens with probability $e^{-\sqrt{n}}$, and the score of the player is always between $0$ and $\frac{\ln n}{2}$, the expected score of the player decreases by at most $e^{-\sqrt{n}} \cdot \frac{\ln n}{2}$, and is thus still strictly higher than $1$.

However, this is impossible:
After a random permutation, $X_1, \ldots, X_n$ are still $n$ independent samples from $\D$. Then, whenever the player accepts an option $i$ after testing it at some $\theta_i < \frac{\ln n}{2}$, the expected score is upper bounded by
\begin{align*}
    \Ex{X_i \sim \D}{X_i - \theta_i | X_i > \theta_i}
&=  \Ex{X \sim \Exp(1)}{\left.\min\left\{X, \frac{\ln n}{2}\right\} - \theta_i \right| \min\left\{X, \frac{\ln n}{2}\right\} > \theta_i}\\
&\le\Ex{X \sim \Exp(1)}{X - \theta_i | X > \theta_i}
=   1.
\end{align*}
This gives a contradiction. 

Therefore, when $a_1, a_2, \ldots, a_n$ are drawn i.i.d.\ from $\D$, the expected optimum is $\Omega(\log n)$ whereas the expected score of any algorithm is $O(1)$. By an averaging argument, for every algorithm there exists (deterministic) $a_1, a_2, \ldots, a_n$ on which the player's competitive ratio is $\Omega(\log n)$.
\end{proof}

\section{A Single-Sample Algorithm for the Prophet Setting}\label{sec:prophet-single-sample}
Suppose that, in the prophet setting of online pen testing, the distributions $\D_1, \D_2, \ldots, \D_n$ from which $(X_i)_{i=1}^n$ is drawn are unknown, and we only get to learn them by drawing a few samples from each $\D_i$. How many samples are sufficient for the player to be still $O(\log n)$-competitive?

Naturally, we might want to simulate the algorithm drawn from Section~\ref{sec:prophet-general} using samples from $\D_1, \D_2, \ldots, \D_n$. Recall that the algorithm only requires a few quantiles of the $n$ distributions, namely $\tau^{(i)}_{\theta}$ for all $i \in [n]$ and $\theta \in \{1/2, 1/4, \ldots, 1/n\}$. Calculating these quantiles (even approximately) from samples, however, requires $\Omega(n)$ samples from each $\D_i$.

Despite this, we show that a single sample from each $\D_i$ is sufficient, thus proving Theorem~\ref{thm:prophet-single-sample}. We also emphasize that unlike in Sections \ref{sec:prophet-iid}~and~\ref{sec:prophet-general}, the following proof allows non-continuous distributions, i.e., $\D_i$ may have point masses.

\begin{proof}[Proof of Theorem~\ref{thm:prophet-single-sample}]
    Let random variable $\Xmax$ denote the maximum of $X_1, \ldots, X_n$ when they are drawn independently from $\D_1, \ldots, \D_n$. The cumulative distribution function of $\Xmax$, $F(x) \coloneqq \pr{X \sim \D}{\Xmax \le x}$, is right-continuous. Thus, it is valid to define $A \ge 0$ as the smallest number such that $\pr{X \sim \D}{\Xmax \le A} = F(x) \ge 1/3$. Similarly, we can define $B \ge 0$ as the largest number such that $\pr{X \sim \D}{\Xmax \ge B} \ge 1/3$. We can verify that $A \le B$, and
    \[
        \pr{X \sim \D}{A \le \Xmax \le B}
    \ge 1 - \pr{X \sim \D}{\Xmax < A} - \pr{X \sim \D}{\Xmax > B}
    \ge 1 - \frac{1}{3} - \frac{1}{3}
    =   \frac{1}{3}.
    \]
    
    Let $\Xmaxhat$ denote the maximum among $\hat X_1, \hat X_2, \ldots, \hat X_n$, where $\hat X_i$ is the sample from $\D_i$ that is provided to the player. Clearly, $\Xmaxhat$ follows the same distribution as $\Xmax$. Thus, $\Xmaxhat \in [A, B]$ holds with probability at least $1/3$, and we condition on this event in the following.
    
    We will give two different algorithms that achieve an $\Omega\left(\frac{1}{\log n}\right)$ fraction of $\Xmaxhat$ and $\Ex{X \sim \D}{\Xmax} - \Xmaxhat$, respectively.
    
    \paragraph{The first algorithm.} Assuming that $\Xmaxhat \le B$, we have
    \[
        \pr{X \sim \D}{\Xmaxhat \le \Xmax}
    \ge \pr{X \sim \D}{B \le \Xmax}
    \ge 1/3.
    \]
    Then, assuming that $\Xmaxhat \le \Xmax$ holds, we may view $X_1, X_2, \ldots, X_n$ as an instance of the secretary setting under an arbitrary arrival order, and $\Xmaxhat$ can be viewed as a ``hint'' on the maximum value that satisfies the precondition of Lemma~\ref{lemma:secretary-arbitrary-approx-opt}. Thus, if we run the algorithm from the lemma with $\Xmaxhat$, we obtain an expected score of $\frac{1}{3}\cdot\frac{\Xmaxhat}{O(\log n)} = \frac{\Xmaxhat}{O(\log n)}$.

    \paragraph{The second algorithm.} The second algorithm is the single-threshold algorithm with $\theta = \Xmaxhat$, and its expected score is given by
    \[
        \sum_{i=1}^n\left(\prod_{j=1}^{i-1}\pr{X_j \sim \D_j}{X_j \le \Xmaxhat}\right)\cdot\pr{X_i \sim \D_i}{X_i > \Xmaxhat}\cdot\Ex{X_i \sim \D_i}{X_i - \Xmaxhat|X_i > \Xmaxhat}.
    \]
    The first factor, $\prod_{j=1}^{i-1}\pr{X_j \sim \D_j}{X_j \le \Xmaxhat}$, is at least 
    \[\pr{X \sim \D}{\Xmax \le \Xmaxhat} \ge \pr{X \sim \D}{\Xmax \le A} \ge 1/3.\]
    The product of the other two factors is equal to
    \[\Ex{X_i \sim \D_i}{(X_i - \Xmaxhat)\cdot\1{X_i > \Xmaxhat}} = \Ex{X_i \sim \D_i}{\max\{X_i - \Xmaxhat, 0\}}.
    \]
    Therefore, the expected score is lower bounded by
    \begin{align*}
        \frac{1}{3}\sum_{i=1}^n\Ex{X_i \sim \D_i}{\max\{X_i - \Xmaxhat, 0\}}
    &=  \frac{1}{3}\Ex{X_i \sim \D_i}{\sum_{i=1}^n\max\{X_i - \Xmaxhat, 0\}}\\
    &\ge\frac{1}{3}\Ex{X \sim \D}{\max_{i \in [n]}\max\{X_i - \Xmaxhat, 0\}}\\
    &=  \frac{1}{3}\Ex{X \sim \D}{\max\{\Xmax - \Xmaxhat, 0\}}
    \ge \frac{\Ex{X \sim \D}{\Xmax} - \Xmaxhat}{3}.
    \end{align*}
    
    Therefore, conditioning on each $\Xmaxhat \in [A, B]$, running one of the two algorithms uniformly at random gives an expected score of
    \[
        \frac{1}{2}\left[\frac{\Xmaxhat}{O(\log n)} + \frac{\Ex{X \sim \D}{\Xmax} - \Xmaxhat}{3}\right]
    =   \frac{\Ex{X \sim \D}{\Xmax}}{O(\log n)}.
    \]
    The overall expected score is then $\frac{1}{3} \cdot \frac{\Ex{X \sim \D}{\Xmax}}{O(\log n)}$, so the algorithm is $O(\log n)$-competitive.
\end{proof}

\bibliographystyle{alpha}
\bibliography{main}

\appendix

\section{Proof of Fact~\ref{fact:lower}}\label{sec:omitted-prophet-iid}
We finish the proof of Fact~\ref{fact:lower} by calculating the expected maximum of $X_1, X_2, \ldots, X_n$, which are drawn i.i.d.\ from the exponential distribution $\D$. Define $\Xmax \coloneqq \max_{i \in [n]}X_i$. We have
\[
    \Ex{X_1, \ldots, X_n\sim \D}{\Xmax}
=   \int_{0}^{+\infty}\pr{X_1, \ldots, X_n\sim \D}{\Xmax \ge u}~\rmd u\\
=  \int_{0}^{+\infty}[1 - (1 - e^{-u})^n]~\rmd u.
\]
Expanding $(1-e^{-u})^n$ and interchanging the summation and integration gives
\[
    \int_{0}^{+\infty}\left[\sum_{k=1}^n\binom{n}{k}(-1)^{k+1}e^{-ku}\right]~\rmd u
=   \sum_{k=1}^n\binom{n}{k}(-1)^{k+1}\int_{0}^{+\infty}e^{-ku}~\rmd u
=   \sum_{k=1}^n\frac{(-1)^{k+1}}{k}\binom{n}{k}.
\]
Plugging the identity $\binom{n}{k} = \sum_{j=0}^{n-1}\binom{j}{k-1}$ into the above gives
\[
    \sum_{k=1}^n\sum_{j=0}^{n-1}\frac{(-1)^{k+1}}{k}\binom{j}{k-1}
=   \sum_{j=0}^{n-1}\frac{1}{j+1}\sum_{k=1}^n(-1)^{k+1}\binom{j+1}{k}
=  \sum_{j=0}^{n-1}\frac{1}{j+1}
=   H_n,
\]
where the first step applies $\binom{j+1}{k} = \frac{j+1}{k}\binom{j}{k-1}$, and the second step follows from
\[
    \sum_{k=1}^n(-1)^{k+1}\binom{j+1}{k}
=   -\sum_{k=1}^{j+1}(-1)^{k}\binom{j+1}{k}
=   (-1)^0\binom{j+1}{0} - [1 + (-1)]^{j+1}
=   1.
\]
This proves $\Ex{X_1, \ldots, X_n\sim \D}{\Xmax} = H_n$.

\section{Proof of Lemma~\ref{lemma:secretary-random-approx-opt}}\label{sec:omitted-secretary-upper}
In the following, we restate and prove Lemma~\ref{lemma:secretary-random-approx-opt}.

\vspace{6pt}

\noindent\textbf{Lemma~\ref{lemma:secretary-random-approx-opt}}~\textit{
In the secretary setting under random order, there is an algorithm that, given any $\Xmaxhat$ that lies in $[0, a_{[1]}]$, achieves a score of at least $\frac{\Xmaxhat}{O(\log n)}$ in expectation.
}

\begin{proof}
When $\Xmaxhat = 0$, there is nothing to prove, so we assume $\Xmaxhat > 0$ in the following. Again, we set $k = \lfloor\log_2n\rfloor + 2 = O(\log n)$ and pick $j$ from $[k - 1]$ uniformly at random. The only difference is that we run the single-threshold algorithm at threshold $\theta = \frac{j-1}{k} \cdot \Xmaxhat$, as we do not know $a_{[1]}$.

For each $j \in [k]$, we define $n_{\ge j} \coloneqq |\{i \in [n]: a_i > \frac{j-1}{k}\Xmaxhat\}|$. Since $a_{[1]} \ge \Xmaxhat > \frac{k-1}{k}\Xmaxhat$, we have $n_{\ge k} \ge 1$. Furthermore, we clearly have $n_{\ge 1} \le n$. Conditioning on the choice of $j$, exactly $n_{\ge j}$ options could pass the test at $\theta = \frac{j-1}{k}\cdot\Xmaxhat$, and the option that we accept is uniformly distributed among them. So, we achieve a score of $\ge \frac{\Xmaxhat}{k}$ with probability $\ge \frac{n_{\ge j+1}}{n_{\ge j}}$. Finally, averaging over the choice of $j \in [k - 1]$ gives an expected score of at least
\[
    \frac{\Xmaxhat}{k}\cdot\frac{1}{k-1}\sum_{j=1}^{k-1}\frac{n_{\ge j+1}}{n_{\ge j}}
\ge \frac{\Xmaxhat}{k}\cdot\left(\frac{n_{\ge k}}{n_{\ge 1}}\right)^{\frac{1}{k-1}}
\ge \frac{a_{[1]}}{k}\cdot n^{-\frac{1}{k-1}}
=   \frac{\Xmaxhat}{O(\log n)}.
\]
\end{proof}

\section{A Slightly Improved Bound for IID Prophet Setting}\label{sec:constants}
We refine the algorithm for the i.i.d.\ prophet setting in Section~\ref{sec:prophet-iid} to give an $(e + o(1))\ln n$-competitive algorithm, as we claimed in Section~\ref{sec:discussion}.

\begin{proof}[Proof of Theorem~\ref{thm:prophet} (i.i.d.\ case with a better constant)]
We will follow the same approach as in the proof from Section~\ref{sec:prophet-iid}: Pick integer $k \ge 1$ and $1 = \alpha_0 > \alpha_1 > \alpha_2 > \cdots > \alpha_k > 0$. We will draw $\theta$ randomly from some distribution over $\{\tau_{\alpha_0}, \tau_{\alpha_1}, \ldots, \tau_{\alpha_k}\}$, and run the single-threshold algorithm at $\theta$.

We first lower bound the expected score of the algorithm conditioning on using each threshold $\tau_{\alpha_j}$. For $j < k$, using threshold $\tau_{\alpha_j}$ gives an expected score of
\begin{align*}
    [1 - (1 - \alpha_j)^n]\cdot\Ex{X \sim \D}{X - \tau_{\alpha_j}|X > \tau_{\alpha_j}}
&\ge[1 - (1 - \alpha_j)^n]\cdot\pr{X \sim \D}{X > \tau_{\alpha_{j+1}} | X > \tau_{\alpha_j}}\cdot (\tau_{\alpha_{j+1}} - \tau_{\alpha_j})\\
&=  \frac{\alpha_{j+1}}{\alpha_j}[1 - (1 - \alpha_j)^n]\cdot(\tau_{\alpha_{j+1}} - \tau_{\alpha_j})\\
&\ge \frac{\alpha_{j+1}}{\alpha_j}[1 - (1 - \alpha_k)^n]\cdot(\tau_{\alpha_{j+1}} - \tau_{\alpha_j}),
\end{align*}
while threshold $\tau_{\alpha_k}$ gives
\[
[1 - (1 - \alpha_k)^n]\cdot\Ex{X \sim \D}{X - \tau_{\alpha_k}|X > \tau_{\alpha_k}}
=   \frac{1}{\alpha_k}[1 - (1 - \alpha_k)^n]\cdot\Ex{X\sim\D}{\max\{X - \tau_{\alpha_k}, 0\}}.
\]
On the other hand, Lemma~\ref{lemma:prophet-optimum} upper bounds $\Ex{X_1, X_2, \ldots, X_n \sim \D}{\max\{X_1, X_2, \ldots, X_n\}}$ by
\[
    \tau_{\alpha_k} + n\cdot\Ex{X \sim \D}{\max\{X - \tau_{\alpha_k}, 0\}}
=   \sum_{j=0}^{k-1}(\tau_{\alpha_{j+1}} - \tau_{\alpha_j}) + n\cdot\Ex{X \sim \D}{\max\{X - \tau_{\alpha_k}, 0\}}.
\]

Define $C_j \coloneqq [1 - (1 - \alpha_k)^n]\cdot\frac{\alpha_{j+1}}{\alpha_j}$ for $j < k$ and $C_k \coloneqq [1 - (1 - \alpha_k)^n]\cdot\frac{1}{n\alpha_k}$. The three equations above together give a $\gamma$-competitive algorithm, where $\gamma \coloneqq \sum_{j=0}^{k}\frac{1}{C_k}$: If we set $\theta$ to $\tau_{\alpha_j}$ with probability $\frac{1/C_j}{\gamma}$, our expected score is lower bounded by
\begin{align*}
    &~\sum_{j=0}^{k-1}\frac{1/C_j}{\gamma}\cdot C_j(\tau_{\alpha_{j+1}} - \tau_{\alpha_j}) + \frac{1/C_k}{\gamma}\cdot nC_k\cdot\Ex{X \sim \D}{\max\{X - \tau_{\alpha_k}, 0\}}\\
=   &~\frac{1}{\gamma}\left[\sum_{j=0}^{k-1}(\tau_{\alpha_{j+1}} - \tau_{\alpha_j}) + n\cdot\Ex{X \sim \D}{\max\{X - \tau_{\alpha_k}, 0\}}\right],
\end{align*}
which is at least a $1/\gamma$ fraction of the expected maximum. Thus, it remains to pick $k$ and $\alpha_0, \alpha_1, \ldots, \alpha_k$ to minimize the competitive ratio $\gamma = \frac{1}{1 - (1 - \alpha_k)^n}\left(\frac{\alpha_0}{\alpha_1} + \frac{\alpha_1}{\alpha_2} + \cdots + \frac{\alpha_{k-1}}{\alpha_k} + n\alpha_k\right)$.

If we fix $\alpha_k = x \in (0, 1)$, the optimal choice of $\alpha_0$ through $\alpha_{k-1}$ are $\alpha_j = x^{j/k}$, and the resulting $\gamma$ can be written as
\[
    \gamma = \frac{kx^{-1/k} + nx}{1 - (1 - x)^n}
\le \frac{kx^{-1/k} + nx}{1 - e^{-nx}}.
\]
If we further choose $x = \frac{\sqrt{\ln n}}{n}$ and $k = \left\lceil\ln\frac{1}{x}\right\rceil \le \ln\frac{1}{x} + 1$, we have
\[
    \gamma
\le \frac{e(\ln\frac{1}{x}+1) + nx}{1 - e^{-nx}}
\le \frac{e\ln n + e + \sqrt{\ln n}}{1 - e^{-\sqrt{\ln n}}}
=   \frac{(1 + o(1))e\ln n}{1 - o(1)}
=   (e + o(1))\ln n.
\]
\end{proof}

\end{document}